\documentclass{lmcs}
\pdfoutput=1

\usepackage{lastpage}
\lmcsdoi{17}{1}{8}
\lmcsheading{}{\pageref{LastPage}}{}{}%
{Sep.~20,~2019}{Feb.~01,~2021}{}

\usepackage[utf8]{inputenc}
\usepackage[T1]{fontenc}
\usepackage{algorithm}
\usepackage[noend]{algpseudocode}
\usepackage{amsmath}
\usepackage{amssymb}
\usepackage{amsthm}
\usepackage{float}
\usepackage{graphicx}
\usepackage{booktabs}
\usepackage{tabularx}

\usepackage{microtype}
\usepackage{relsize}
\usepackage{stmaryrd}
\usepackage{tikz}
\usepackage{todonotes}
\usepackage{thm-restate}
\usepackage{bm}
\usepackage{thmtools}
\usepackage{xspace}

\theoremstyle{plain}\newtheorem{claim}[thm]{Claim}

\DeclareFontFamily{U}{mathx}{\hyphenchar\font45}
\DeclareFontShape{U}{mathx}{m}{n}{<-> mathx10}{}
\DeclareSymbolFont{mathx}{U}{mathx}{m}{n}
\DeclareMathAccent{\widebar}{0}{mathx}{"73}

\usetikzlibrary{arrows}

\let\oldReturn\Return
\renewcommand{\Return}{\State\oldReturn}
\renewcommand{\gets}{:=}

\newcommand{\powerset}{\mathcal{P}}

\newcommand{\ie}{\emph{i.e.}}
\newcommand{\eg}{\emph{e.g.}}
\newcommand{\viz}{\emph{viz.}}

\newcommand{\failures}{\mathsf{failures}}
\newcommand{\failuresbot}{\mathsf{failures}_\bot}
\newcommand{\divergences}{\mathsf{divergences}}
\newcommand{\divergencesbot}{\mathsf{divergences}_{\mathsf{min}}}
\newcommand{\refinedbytrace}{\sqsubseteq_\mathsf{tr}}
\newcommand{\refinedbyfdr}{\sqsubseteq_\mathsf{fdr}}
\newcommand{\refinedbysfr}{\sqsubseteq_\mathsf{sfr}}

\newcommand{\Antichain}{\mathcal{A}}
\newcommand{\chaininsert}{\Cup}
\newcommand{\chainin}{\Subset}
\newcommand{\chainnotin}{\not\Subset}
\newcommand{\chainsetin}{\Subset^\forall}
\newcommand{\statesubset}{\leq}
\newcommand{\statensubset}{\nleq}

\newcommand{\init}{\ensuremath{\iota}\xspace}
\newcommand{\lts}{\mathcal{L}}
\newcommand{\states}{\mathit{S}}
\newcommand{\actions}{\mathit{Act}}
\newcommand{\actionstau}{\mathit{Act_\tau}}
\newcommand{\transitions}{\rightarrow}

\newcommand{\transition}[1]{\xrightarrow{#1}\!\!\!\!\!\!\rightarrow}
\newcommand{\weaktransition}[1]{\overset{#1}{\Longrightarrow}}
\newcommand{\tracetransition}[1]{\ensuremath{\xrightarrow{#1}\!\!\!\!\!\rightarrow}}

\newcommand{\refusals}{\mathsf{refusals}}
\newcommand{\reachable}{\mathsf{reachable}}
\newcommand{\enabled}{\mathsf{enabled}}
\newcommand{\stable}{\mathsf{stable}}
\newcommand{\diverges}[1]{#1\raisebox{1pt}{\(\Uparrow\)}}

\newcommand{\normalize}{\ensuremath{\mathsf{norm}_\mathsf{fdr}}}
\newcommand{\normalizesfr}{\ensuremath{\mathsf{norm}}}
\newcommand{\lscott}{\ensuremath{[\![}}
\newcommand{\rscott}{\ensuremath{]\!]}}
\newcommand{\tostates}[1]{\ensuremath{\lscott #1\mkern1mu\rscott}}

\newcommand{\product}{\ltimes}

\newcommand{\emptytrace}{\epsilon}
\newcommand{\trace}{\sigma}
\newcommand{\weaktrace}{\rho}
\newcommand{\traces}{\mathsf{traces}}
\newcommand{\weaktraces}{\mathsf{weaktraces}}
\newcommand{\length}[1]{|#1|}

\newcommand{\refinesnew}[1]{\ensuremath{\textsc{refines-#1}_\textsc{new}}}
\newcommand{\antichain}{\mathit{antichain}}
\newcommand{\impl}{\mathit{impl}}
\newcommand{\spec}{\mathit{spec}}
\newcommand{\working}{\mathit{working}}

\newcommand{\Done}{\mathit{done}}
\newcommand{\Dist}{\mathsf{Dist}}
\newcommand{\FDR}{\mathsf{FDR}}
\newcommand{\witnesses}{\mathcal{F}}

\newcommand{\functionorder}{\mathcal{O}}

\newcommand{\cpp}{\texttt{C++}}
\newcommand{\outofmemory}{\ensuremath{\dagger}}

\newcommand{\branchingbisim}{\leftrightarroweq_{db}}

\newcommand{\request}{\textsc{req}}


\begin{document}

\title{Correct and Efficient Antichain Algorithms for Refinement Checking}
\author[M.~Laveaux]{Maurice Laveaux}
\author[J.\,F.~Groote]{Jan Friso Groote}
\author[T.\,A.\,C.~Willemse]{Tim A.C. Willemse}
\address{Eindhoven University of Technology. De Groene Loper 5, 5612 AE, Eindhoven, The Netherlands}
\email{m.laveaux@tue.nl, j.f.groote@tue.nl, t.a.c.willemse@tue.nl}



\begin{abstract}
The notion of refinement plays an important role in software engineering.
It is the basis of a stepwise development methodology in which the correctness of a system can be established by proving, or computing, that a system refines its specification.
Wang \emph{et al}.\ describe algorithms based on antichains for efficiently deciding trace refinement, stable failures refinement and failures-divergences refinement. 
We identify several issues pertaining to the soundness and performance in these algorithms and propose new, correct, antichain-based algorithms.
Using a number of experiments we show that our algorithms outperform the original ones in terms of running time and memory usage.
Furthermore, we show that additional run time improvements can be obtained by applying divergence-preserving branching bisimulation minimisation.
\end{abstract}

\maketitle

\section{Introduction}

Refinement is often an integral part of a mature engineering methodology 
for designing a (software) system in a stepwise manner.
It allows one to start from a high-level specification that describes the permitted and desired behaviours of a system and arrive at a detailed implementation that behaves according to this specification.
While in many settings, refinement is often used rather informally, it forms the mathematical cornerstone in the theoretical development of the process algebra CSP (Communicating Sequential Processes) by Hoare~\cite{Hoare85,Roscoe94,Roscoe10}. 

This formal view on refinement---as a mathematical relation between a specification and its implementa\-tion---has been used successfully in industrial settings~\cite{GomesB16,Gibson-Robinson17}, and it has been incorporated in commercial Formal Model-Driven Engineering tools such as \emph{Dezyne}~\cite{BeusekomGHHWWW17}.
In such settings there are a variety of refinement relations, each with their own properties.
In particular, each notion of refinement offers specific guarantees on the (types of) behavioural properties of the specification that carry over to correct implementations. 
For instance, \emph{trace refinement}~\cite{Roscoe10} only preserves safety properties.
The---arguably---most prominent refinement relations for the theory of CSP are the \emph{stable failures refinement}~\cite{BergstraKO87,Roscoe10} and \emph{failures-divergences refinement}~\cite{Roscoe10}.
All three refinement relations are implemented in the FDR~\cite{Gibson-RobinsonABR14,Gibson-Robinson16} tool for specifying and analysing CSP processes.

Trace refinement, stable failures refinement and failures-divergences refinement are computationally hard problems; deciding whether there is a refinement relation between an implementation and a specification, both represented by CSP processes or labelled transition systems, is PSPACE-hard~\cite{KanellakisS90}.
In practice, however, tools such as FDR are able to work with quite large state spaces. 
The basic algorithm for deciding a trace refinement, stable failures refinement or a failures-divergences refinement between implementation and specification relies on a \emph{normalisation} of the specification.
This normalisation is achieved by a subset construction that is used to obtain a deterministic transition system which represents the specification. 

As observed in~\cite{WangS0LDWL12} and inspired by successes reported, \eg, in~\cite{AbdullaCHMV10,DoyenR10,WulfDHR06}, \emph{antichain} techniques can be exploited to improve on the performance of refinement checking algorithms.
Unfortunately, a closer inspection of the results and algorithms in~\cite{WangS0LDWL12} reveals several issues. 
First, the definitions of stable failures refinement and failures-divergences refinement used in~\cite{WangS0LDWL12} do not match the definitions of~\cite{BergstraKO87,Roscoe10}, nor do they seem to match known relations from the literature~\cite{Glabbeek19}.

Second, as we demonstrate in Example~\ref{ex:incorrect} in this paper, the results~\cite[Theorems~2~and~3]{WangS0LDWL12} claiming correctness of their algorithms for deciding both non-standard refinement relations are incorrect. 
We do note that their algorithm for checking trace refinement is correct, and their algorithm for checking stable failures refinement correctly decides the refinement relation defined by~\cite{BergstraKO87,Roscoe10}. 

Third, unlike claimed by the authors, the algorithms of~\cite{WangS0LDWL12} violate the antichain property as we demonstrate in Example~\ref{example:violation}.
Fourth, their algorithms suffer from severely degraded performance due to sub-optimal decisions made when designing the algorithms, leading to an overhead of a factor $|\actions|\cdot |\states|$, where $\actions$ is the set of actions and $S$ the set of states of the implementation, as we show in Example~\ref{example:badcase}. 
This factor is even greater, \viz~$|\actions|^{|\states|}$, when using a FIFO (first in, first out) queue to realise a breadth-first search strategy instead of the stack used for the depth-first search.
Note that there are compelling reasons for using a breadth-first strategy~\cite{Roscoe94}; \eg, the conciseness of counterexamples to refinement.

Our contributions are the following.
Apart from pointing out the issues in~\cite{WangS0LDWL12}, we propose new antichain-based algorithms for deciding trace refinement, stable failures refinement and failures-divergences refinement and we prove their correctness.
We compare the performance of the trace refinement algorithm and the stable failures refinement algorithm of~\cite{WangS0LDWL12} to ours.
Due to the flaw in their algorithm for deciding failures-divergences refinement, a comparison of this refinement relation makes little sense. 
Our results indicate a small improvement in run time performance for practical models when using depth-first search, whereas our experiments using breadth-first search illustrate that decision problems intractable using the algorithm of~\cite{WangS0LDWL12} generally become quite easy using our algorithm.
Finally, we show that divergence-preserving branching bisimulation~\cite{Glabbeek93,GlabbeekLT09} minimisation preserves the desired refinement checking relations and that applying this minimisation as a preprocessing step can yield significant run time improvements.

The current paper is based on~\cite{LaveauxGW19}, but extends it in several respects.
First, in addition to stable failures refinement and failures-divergences refinement, the current exposition also includes a treatment of trace refinement.
Second, we include detailed proofs of correctness of the main claims in~\cite{LaveauxGW19}, and we include several results that are needed to support these claims, but which are also of value on their own.
Some of the more straightforward proofs have been deferred to the appendix.
Third, we expand further on our experimental results.
In particular, we have included more performance metrics, to help explain the observed performance improvements of our algorithms over those of~\cite{WangS0LDWL12},  and we have included a section on using divergence-preserving branching bisimulation for preprocessing.

\paragraph{\bf Outline}
In Section~\ref{section:preliminaries} the preliminaries of labelled transition systems and the refinement relations are defined. 
In Section~\ref{section:refinement_checking} a general procedure for checking refinement relations is described.
In Section~\ref{section:original-algorithm} the antichain-based algorithms of~\cite{WangS0LDWL12} are presented and their issues are described in detail.
In Section~\ref{section:improved-algorithm} the improved antichain algorithms are presented and their correctness is shown. 
Finally, in Section~\ref{section:experiments} an experimental evaluation is conducted to show the effectiveness of these changes in practice, followed by the evaluation of applying divergence-preserving branching bisimulation minimisation as a preprocessing step.

\section{Preliminaries}\label{section:preliminaries}

In this section the preliminaries of labelled transition systems and the considered refinement relations are introduced.
We follow the standard conventions, notation and definitions of~\cite{Brookes85, Roscoe10, Glabbeek17}.

\subsection{Labelled Transition Systems}

Let $\actions$ be a finite set of actions that does not contain the constant $\tau$, which models \emph{internal} actions, and let $\actionstau$ be equal to $\actions \cup \{\tau\}$.

\begin{defi}
  A labelled transition system $\lts$ is a tuple $(\states, \init, \transitions)$ where $\states$ is a set of states; $\init \in \states$ is an initial state and $\transitions\,\subseteq \states \times \actionstau \times \states$ is a labelled transition relation.
\end{defi}

\noindent We depict labelled transition systems as edge-labelled directed graphs, where vertices represent states and the labelled edges between vertices represent the transitions. 
An incoming arrow with no starting state and no action indicates the initial state.
We use the initial state to refer to a depicted LTS.

For the remainder of this section, we assume that $\lts = (\states, \init, \transitions)$ is an arbitrary LTS.
We adopt the following conventions and notation.
Typically, we use symbols $s, t, u$ to denote states, $U, V$ to denote sets of states and $a$ to denote actions. 
A transition $(s, a, t) \in\, \transitions$ is also written as $s \transition{a} t$.
The set of \emph{enabled} actions of state $s$ is defined as $\enabled(s) = \{a \in \actionstau \mid \exists t \in \states : s \transition{a} t\}$.

A sequence is denoted by concatenation, \ie, $a_0\,a_1\,\cdots\,a_{n-1}$ where $a_i \in \actionstau$ for all $0 \leq i < n$ is a sequence of actions and $\actionstau^*$ indicates the set of all finite sequences of actions.
We use $\trace$ and $\weaktrace$ to denote a sequence of actions, where $\weaktrace$ typically does not contain $\tau$.
The \emph{length} of a sequence, denoted as $\length{a_0\,a_1\,\cdots\,a_{n-1}}$, is equal to $n$.
Finally, we say that any sequence $a_0\,a_1\,\cdots\,a_k$ such that $k \leq n - 1$ is a \emph{prefix} of a sequence $a_0\,a_1\,\cdots\,a_{n-1}$.
A prefix is \emph{strict} whenever its length is strictly smaller than that of the sequence itself.

The transition relation of an LTS is generalised to sequences of actions as follows: $s \tracetransition{\emptytrace} t$ holds iff $s = t$, and $s \tracetransition{\trace\,a} t$ holds iff there is a state $u$ such that $s \tracetransition{\trace} u$ and $u \transition{a} t$.
The \emph{weak transition} relation $\weaktransition{}\,\subseteq \states \times \actions^* \times \states$ is the smallest relation satisfying:  
\begin{itemize}
	\item $s \weaktransition{\emptytrace} s$, and
	
	\item $s \weaktransition{\emptytrace} t$ if $s \transition{\tau} t$, and
	 
  \item $s \weaktransition{a} t$ if $s \transition{a} t$ for $a \in \actions$, and
    
  \item $s \weaktransition{\weaktrace\,\trace} t$ if there is a state $u$ such that $s \weaktransition{\weaktrace} u$ and $u \weaktransition{\trace} t$.
\end{itemize}   

\begin{defi}
  Traces, weak traces and reachable states are defined as follows:

  \begin{itemize}
    \item 
      The \emph{traces} starting in state $s$ are defined as $\traces(s) = \{\trace \in \actionstau^* \mid \exists t \in \states : s \tracetransition{\trace} t\}$.
      We define $\traces(\lts)$ to be $\traces(\init)$.

    \item    
      The \emph{weak traces} starting in state $s$ are defined as $\weaktraces(s) = \{\weaktrace \in \actions^* \mid \exists t \in \states : s \weaktransition{\weaktrace} t\}$.
      We define $\weaktraces(\lts)$ to be $\weaktraces(\init)$.

    \item the set of states, \emph{reachable} from $s$ is defined as $\reachable(s) = \{t \in \states \mid \exists \trace \in \actionstau^* : s \tracetransition{\trace} t\}$.
      We define $\reachable(\lts)$ to be $\reachable(\init)$.
  \end{itemize}

\end{defi}

\begin{defi}
  Labelled transition system $\lts$ is:

  \begin{itemize}
    \item \emph{deterministic} if and only if for all states $s, t, u$ and actions $a \in \actionstau$ if there are transitions $s \transition{a} t$ and $s \transition{a} u$ then $t = u$.

    \item \emph{concrete} if it does not contain transitions labelled with $\tau$, \ie, for all states $s$ it holds that $\tau \notin \enabled(s)$.
    
    \item \emph{universal} if and only if for all states $s$ it holds that $\enabled(s) = \actions$.
  \end{itemize} 
\end{defi}

\begin{lem}\label{lemma:deterministic}
  Let $\lts$ be a deterministic LTS.
  For all sequences $\trace \in \actionstau^*$ and states $s, t, u$, if $s \tracetransition{\trace} t$ and $s \tracetransition{\trace} u$ then $t = u$.
\end{lem}

\noindent The models underlying the CSP process algebra~\cite{Hoare85,Roscoe10} build on observations of \emph{weak traces}, \emph{failures} and \emph{divergences}.
A weak trace observation records the visible actions that occur when performing an experiment on the system.
A failure is a combination of a set of actions that a system observably refuses and a weak trace experiment on the system that leads to the observation of the refusals.
A refusal can only be observed when the system has \emph{stabilised}, meaning that it can no longer perform internal behaviour.
A \emph{divergence} can be understood as the potential inability of the system to stabilise, which can happen when the system engages in an infinite sequence of $\tau$-actions after performing an experiment on the system.

\begin{defi}[Refusals]\label{def:refusals}
  A state $s$ is \emph{stable}, denoted by $\stable(s)$, if and only if $\tau \notin \enabled(s)$.
  For a stable state $s$,
  the \emph{refusals} of $s$ are defined as $\refusals(s) = \powerset(\actions \setminus \enabled(s))$.
  For a set of states $U \subseteq S$ its refusals are defined as $\refusals(U) = \{ X \subseteq \actions \mid \exists s \in U : \stable(s) \land X \in \refusals(s)\}$.
\end{defi}
  
\noindent Formally, a state $s$ is \emph{diverging}, denoted by the predicate $\diverges{s}$, if and only if there is an infinite sequence of states $s \transition{\tau} s_1 \transition{\tau} s_2 \transition{\tau} \cdots$.
For a set of states $U $, we write $\diverges{U}$, iff $\diverges{s}$ for some state $s \in U$.

\begin{defi}[Divergences]
  The \emph{divergences} of a state $s$ are defined as $\divergences(s) = \{\weaktrace\,\trace \in \actions^* \mid \exists t \in \states : (s \weaktransition{\weaktrace} t \land \diverges{t})\}$.
  We define $\divergences(\lts) = \divergences(\init)$.
\end{defi}

\noindent Observe that a divergence is any weak trace that has a prefix $\weaktrace$ which can reach a diverging state.
This is based on the assumption that divergences lead to \emph{chaos}.
In theories such as CSP, in which divergences are considered chaotic, chaos obscures all information about the behaviours involving a diverging state; we refer to this as obscuring \emph{post-divergences details}.

\begin{defi}[Stable failures]\label{def:failures}
  The set of all \emph{stable failures} of a state $s$ is defined as $\failures(s) = \{(\weaktrace, X) \in \actions^* \times \powerset(\actions) \mid \exists t \in S : (s \weaktransition{\weaktrace} t \land \stable(t) \land X \in \refusals(t)) \}$.
  The set of failures with post-divergences details \emph{obscured} is defined as $\failuresbot(s) = \failures(s) \cup \{(\weaktrace, X) \in \actions^* \times \powerset(\actions) \mid \weaktrace \in \divergences(s) \}$.
\end{defi}

\noindent We illustrate these concepts by means of an example.
\begin{exa}\label{example:concepts}
  Consider the LTSs $s_0$ and $u_0$ depicted below.    
  \begin{center}
  \scriptsize
  \begin{tikzpicture}[->,node distance=35pt]
   	\tikzstyle{vertex}=[draw=none,minimum size=17pt,inner sep=0pt]
   		
   	\node[vertex] (A1) {$s_0$};
   	\node[vertex] (A1init) [left=10pt of A1] {};
   	\node[vertex] (A1right) [right=.5cm of A1] {$s_1$};
   	\node[vertex] (A1ten) [right= of A1right,yshift=30pt,xshift=-15pt] {$s_2$};
   	\node[vertex] (A1ten2) [left= of A1ten,xshift=15pt] {$s_3$};
   	\node[vertex] (A1twenty) [right= of A1right,yshift=-30pt,xshift=-15pt] {$s_5$};
   		
   	\draw (A1init) edge (A1); 
   	\draw (A1) edge node[above] {$\request$} (A1right);
   	\draw (A1right) edge node[above left] {$\tau$} (A1ten);
   	\draw (A1right) edge node[below left] {$\tau$} (A1twenty);
   	\draw (A1twenty) edge[bend left] node[below left] {20} (A1);
   	\draw (A1ten) edge node[above] {10} (A1ten2);
   	\draw (A1ten2) edge[bend right] node[above left] {10} (A1);	     
       
     \node[vertex] (C1) [right=4cm of A1] {$u_0$};
     \node[vertex] (C1init) [left=10pt of C1] {};
     \node[vertex] (C1right) [right of=C1] {$u_1$};
     \node[vertex] (C1twenty) [right of=C1right] {$u_2$};
     
     \draw (C1init) edge (C1); 
     \draw (C1) edge node[above]{$\request$} (C1right); 
     \draw (C1right) edge node[above]{$20$} (C1twenty);
     \draw (C1twenty) edge[bend left] node[below]{$\tau$} (C1);
     \draw (C1right) edge[loop above] node[above]{$\tau$} (C1right);
  \end{tikzpicture}
  \end{center}    
  \noindent We observe that states $s_0, s_2, s_3, s_5$ and $u_0$ are stable.
  For each of these states we can determine their refusals, \eg, state $s_0$ has the refusals $\{\emptyset, \{10\}, \{20\}, \{10, 20\}\}$ as given by $\refusals(s_0)$.
  Furthermore, we observe that $\request\,20$ is a weak trace of $s_0$ to itself.  
  Consequently, it follows that for example the pairs $(\request\,20, \{10\})$ and $(\request\,20, \{10, 20\})$ are failures of $s_0$.
  None of the states in $s_0$ diverge and as such the corresponding set of divergences are empty and both notions of failures coincide.
  However, for state $u_1$ we can see that $\diverges{u_1}$ holds and therefore $\request$, but also $\request\,10$ is a possible divergence of $u_0$, \ie, $\request\,10 \in \divergences(u_0)$.
  This also means that $(\request\, 10, \{10\}) \in \failuresbot(u_0)$ is a failure of $u_0$ with post-divergences details obscured.\qed
\end{exa}

\noindent The three models of CSP building on the different powers of observation are the \emph{weak trace} model, the \emph{stable failures} model and the \emph{failures-divergences} model.
The refinement relations, induced by these models, are called \emph{trace refinement}, \emph{stable failures refinement} and \emph{failures-divergences refinement} respectively.

\begin{defi}[Refinement]\label{def:refinement}
 	Let $\lts_1$ and $\lts_2$ be two LTSs. 
  \begin{itemize}
    \item $\lts_1$ is refined by $\lts_2$ in trace semantics, denoted by $\lts_1 \refinedbytrace \lts_2$, if and only if  $\weaktraces(\lts_2) \subseteq \weaktraces(\lts_1)$.

    \item $\lts_1$ is refined by $\lts_2$ in stable failures semantics, denoted by $\lts_1 \refinedbysfr \lts_2$, if and only if $\failures(\lts_2) \subseteq \failures(\lts_1)$ and $\weaktraces(\lts_2) \subseteq \weaktraces(\lts_1)$.

    \item $\lts_1$ is refined by $\lts_2$ in failures-divergences semantics, denoted by $\lts_1 \refinedbyfdr \lts_2$, if and only if both $\failuresbot(\lts_2) \subseteq \failuresbot(\lts_1)$ and $\divergences(\lts_2) \subseteq \divergences(\lts_1)$.
  \end{itemize}
\end{defi}

\noindent The LTS that is typically refined is referred to as the \emph{specification}, whereas the LTS that refines the specification is referred to as the \emph{implementation}.

\begin{rem}
  We observe that each refinement relation between LTSs is inverted with respect to the subset relation in its corresponding definition.
  However, in the setting of CSP where these refinement relations are fundamental and have been extensively studied~\cite{Brookes85, Roscoe10, Glabbeek17}, refinement is viewed as an ordering between processes where the process that does not restrict anything, \eg, the process with all failures, is seen as the smallest, least restrictive specification.
\end{rem}
\begin{rem}\label{remark:incorrect}
  The notions defined above appear in different formulations in~\cite{WangS0LDWL12}.
  Their definition of stable failures refinement omits the clause for weak trace inclusion, and their definition of failures-divergences refinement replaces $\failuresbot$ with $\failures$. 
  This yields refinement relations different from the standard ones and neither relation seems to appear in the literature~\cite{Glabbeek19}.
\end{rem}

\noindent We conclude with a small example, illustrating the uses of, and differences between the various refinement relations.

\begin{exa}\label{example:refinement}
  Consider the LTSs $s_0$ and $u_0$ of Example~\ref{example:concepts} again and the LTS $t_0$ depicted in between.
  We now consider $s_0$ to be the specification of a simplified automated teller machine.
\begin{center}
\scriptsize
\begin{tikzpicture}[->,node distance=35pt]
	\tikzstyle{vertex}=[draw=none,minimum size=17pt,inner sep=0pt]
		
	\node[vertex] (A1) {$s_0$};
	\node[vertex] (A1init) [left=10pt of A1] {};
	\node[vertex] (A1right) [right=.5cm of A1] {$s_1$};
	\node[vertex] (A1ten) [right= of A1right,yshift=30pt,xshift=-15pt] {$s_2$};
	\node[vertex] (A1ten2) [left= of A1ten,xshift=15pt] {$s_3$};
	\node[vertex] (A1twenty) [right= of A1right,yshift=-30pt,xshift=-15pt] {$s_5$};
		
	\draw (A1init) edge (A1); 
	\draw (A1) edge node[above] {$\request$} (A1right);
	\draw (A1right) edge node[above left] {$\tau$} (A1ten);
	\draw (A1right) edge node[below left] {$\tau$} (A1twenty);
	\draw (A1twenty) edge[bend left] node[below left] {20} (A1);
	\draw (A1ten) edge node[above] {10} (A1ten2);
	\draw (A1ten2) edge[bend right] node[above left] {10} (A1);
		
	\node[vertex] (B1) [right=4cm of A1] {$t_0$};
	\node[vertex] (B1init) [left=10pt of B1] {};
	\node[vertex] (B1right) [right of=B1] {$t_1$};
	\node[vertex] (B1twenty) [right of=B1right] {$t_2$};
			
	\draw (B1init) edge (B1); 
	\draw (B1) edge node[above]{$\request$} (B1right); 
	\draw (B1right) edge node[above]{$20$} (B1twenty);
  
  \node[vertex] (C1) [right=4cm of B1] {$u_0$};
  \node[vertex] (C1init) [left=10pt of C1] {};
  \node[vertex] (C1right) [right of=C1] {$u_1$};
  \node[vertex] (C1twenty) [right of=C1right] {$u_2$};
  
  \draw (C1init) edge (C1); 
  \draw (C1) edge node[above]{$\request$} (C1right); 
  \draw (C1right) edge node[above]{$20$} (C1twenty);
  \draw (C1twenty) edge[bend left] node[below]{$\tau$} (C1);
  \draw (C1right) edge[loop above] node[above]{$\tau$} (C1right);
\end{tikzpicture}
\end{center} 
  \noindent In the specification $s_0$ the user can first request, by action $\request$, an amount of twenty from the machine.
  The machine can then satisfy this request by either choosing to give twenty directly or by presenting two times ten to the user, which might vary depending on availability within the machine.
  Note that the distinction between user-initiated and response actions is only for the sake of the explanation and is not formally present in the LTS.
      
  An implementation of this specification is \emph{valid} if and only if it refines the specification in the required refinement semantics.  
  Let us consider $t_0$ as a first implementation of this machine.
 The weak traces of $t_0$, consisting of the set $\{\epsilon,\request, \request\,20\}$, are included in the specification and therefore $s_0 \refinedbytrace t_0$.
  Trace refinement is suitable for safety properties; for example the absence of infinite sequences of $10$ or $20$ actions without matching requests can be specified in such semantics.  
  However, trace refinement does not preserve liveness properties.
  In particular, deadlocks are not preserved by trace refinement.
  In contrast, stable failures refinement \emph{does} preserve deadlock freedom.
  For instance, the observation of the failure $(\request\,20, \{\request, 20, 10\})$ of 
  $t_0$ leads to the \emph{deadlocked} state $t_2$, \ie, a state 
  with no outgoing transitions. 
  This failure is not among the failures that can be observed of state $s_0$.
  Consequently, $s_0 \not\refinedbysfr t_0$. 
  
  The second implementation that we consider is $u_0$.
  The self-loop above state $u_1$ might indicate that the machine uses (repeated) polling via a potentially unstable connection to determine whether the user's bank account permits the requested withdrawal. 
  Note that $u_1$ is not a stable state; all failures are thus of the shape $(\request\,20\,\request\,20\ldots, \{10, 20\})$ and these are permitted observations for the given specification $s_0$.
  Therefore, this implementation is a valid stable failures refinement of the given specification, \ie, $s_0 \refinedbysfr u_0$.
  The divergences $\request\,\weaktrace$ for sequences $\weaktrace \in \actions^*$ cause this implementation not to be valid under failures-divergences refinement, \ie, $s_0 \not\refinedbyfdr u_0$.  
  From the perspective of the user, it is indeed questionable whether $u_0$ constitutes a proper implementation, as she may perceive a divergence of the system as a deadlock.\qed
\end{exa}

\noindent 
Note that stable failures refinement is a stronger relation than trace refinement; \ie, whenever $\refinedbysfr$ holds then also necessarily $\refinedbytrace$ holds.
This does not hold the other way around as already shown in Example~\ref{example:refinement} where $s_0 \refinedbytrace t_0$ and $s_0 \not\refinedbysfr t_0$.
Furthermore, $\refinedbyfdr$ is incomparable to $\refinedbytrace$. 
In the preceding example, we  have $u_0 \refinedbyfdr s_0$, but
$u_0 \not\refinedbytrace s_0$ because $\request\,10 \in \weaktraces(s_0)$ and $\request\,10 \not\in \weaktraces(u_0)$.
But we also have $s_0 \refinedbytrace t_0$ and $s_0 \not\refinedbyfdr t_0$, where the latter fails because $(\request\,20, \{\request, 20, 10\}) \in \failuresbot(t_0)$, but $(\request\,20, \{\request, 20, 10\}) \notin \failuresbot(s_0)$.
Similarly, $\refinedbyfdr$ is incomparable to $\refinedbysfr$.
For instance, we have $s_0 \refinedbysfr u_0$ and $s_0 \not\refinedbyfdr u_0$ in Example~\ref{example:refinement},
but also $u_0 \refinedbyfdr s_0$ and $u_0 \not\refinedbysfr s_0$, where for the latter we observe that $(\request, \{10\}) \in \failures(s_0)$ but $(\request, \{10\}) \notin \failures(u_0)$.

\section{Refinement Checking}\label{section:refinement_checking}

In general, the set of weak traces, failures and divergences of an LTS can be infinite. Therefore, checking inclusion of these sets directly is not always viable. 
In~\cite{Roscoe94,Roscoe10}, an algorithm to decide refinement between two labelled transition
systems is sketched.  
As a preprocessing step to this algorithm, all diverging states in both LTSs are marked. 
The algorithm then relies on exploring the product of a \emph{normal form} representation of the specification, \ie, the LTS that is to be refined, and the implementation. 

For each state in this product it checks whether it can locally decide non-refinement of the implementation state with the normal form state.
A state for which non-refinement holds is referred to as a \emph{witness}.
Following~\cite{Roscoe10,WangS0LDWL12} and specifically the terminology of~\cite{Roscoe94},
we formalise the product between LTSs that is explored by the
procedure.

\begin{defi}[Product]\label{def:product} 
  Let $\lts_1 = (\states_1, \init_1, \transitions_1)$ and $\lts_2 = (\states_2, \init_2, \transitions_2)$ be two LTSs.
  The \emph{product} of $\lts_1$ and $\lts_2$, denoted by $\lts_1 \product \lts_2$, is an LTS $(\states, \init, \transitions)$ such that $\states = \states_1 \times \states_2$ and $\init = (\init_1, \init_2)$. The transition relation $\transitions$ is the smallest relation such that for all $s_1, t_1 \in \states_1$, and $s_2, t_2 \in \states_2$ and $a \in \actions$:
  
  \begin{itemize}
    \item If $s_2 \transition{\tau}_2 t_2$ then $(s_1, s_2) \transition{\tau} (s_1, t_2)$.
    
    \item If $s_1 \transition{a}_1 t_1$ and $s_2 \transition{a}_2 t_2$ then $(s_1, s_2) \transition{a} (t_1, t_2)$.
  \end{itemize} 
\end{defi}

\noindent The proposition below relates the behaviours of two LTSs to the behaviours of their product.

\begin{prop}\label{lemma:product_traces}
  Let $\lts_1 = (\states_1, \init_1, \transitions_1)$ and $\lts_2 = (\states_2, \init_2, \transitions_2)$ be two LTSs and let $\lts_1 \product \lts_2 = (\states, \init, \transitions)$.
  For all states $(s_1, t_1), (s_2, t_2) \in \states$ and all sequences $\weaktrace \in \actions^*$ it holds that $(s_1, s_2) \weaktransition{\weaktrace} (t_1, t_2)$ if and only if $s_1 \tracetransition{\weaktrace}_1 t_1$ and $s_2 \weaktransition{\weaktrace}_2 t_2$. 
\end{prop}

\begin{proof}
  First, we can show that the statement holds for the empty sequence by induction on the length of sequences in $\tau^*$.
  Using this we can prove the statement for all sequences in $\actions^*$ by induction on their length using the previous result in the base case.
\end{proof}

\noindent The normal form LTS of a given LTS is obtained using a typical subset construction as is common when determinising a transition system.
A difference between determinisation and normalisation is that the former yields a transition system that preserves and reflects the set of weak traces of the given LTS. This is not the case for the latter, which may add weak traces not present in the original LTS.
We first introduce the normal form LTS of a given LTS that is adequate for reducing the trace refinement and stable failures decision problems to a reachability problem.

\begin{defi}[Normal form]\label{def:normalise_sfr}
  Let $\lts = (\states, \init,  \transitions)$ be an LTS.
  The normal form of $\lts$ is the LTS $\normalizesfr(\lts) = (\states', \init', \transitions')$, where $\states' = \powerset(\states)$, $\init' = \{s \in \states \mid \init \weaktransition{\emptytrace} s \}$ and $\transitions'$ is defined as $U \transition{a}' V$ if and only if $V = \{t \in \states \mid \exists s \in U : s \weaktransition{a} t\}$ for all sets of states $U, V \subseteq \states$ and actions $a \in \actions$.
\end{defi}

\noindent Notice that $\emptyset$ is a state in a normal form LTS. 
Furthermore, the normal form LTS is \emph{deterministic}, \emph{concrete} and \emph{universal}. 

Since a normal form LTS is concrete, all of its states are stable. The states of the original LTS comprising a normal form state may not be stable, however.
When we need to reason about the stability and refusals of the set of states $U$ in the LTS $\lts$ underlying a normal form LTS, rather than the state $U$ of the normal form
LTS, we therefore write $\tostates{U}_\lts$ whenever we wish to stress that we refer to the set of states in $\lts$ that comprise $U$.

The three lemmas stated below relate the set of weak traces of an LTS $\lts$ to the set of traces of $\normalizesfr(\lts)$.

\begin{restatable}{lem}{lemmasubsetsfr}\label{lemma:subset_sfr}
  Let $\lts = (\states, \init, \transitions)$ be an LTS and let $\normalizesfr(\lts) = (\states', \init', \transitions')$.
  For all sequences $\weaktrace \in \actions^*$ and states $U \in \states'$ such that $\init' \tracetransition{\weaktrace}' U$, it holds that $\init \weaktransition{\weaktrace} s$ for all $s \in \tostates{U}_\lts$.
\end{restatable}

\begin{restatable}{lem}{lemmasubsetexistencesfr}\label{lemma:subset_existence_sfr}
  Let $\lts = (\states, \init, \transitions)$ be an LTS and let $\normalizesfr(\lts) = (\states', \init', \transitions')$.
  For all sequences $\weaktrace \in \actions^*$ and for all states $s \in \states$ such that $\init \weaktransition{\weaktrace} s$, there is a state $U \in \states'$ such that $s \in \tostates{U}_\lts$ and $\init' \tracetransition{\weaktrace}' U$.
\end{restatable}

The lemma below clarifies the role of the state $\emptyset$ in $\normalizesfr(\lts)$.

\begin{restatable}{lem}{lemmanotweaktracesfr}\label{lemma:not_weaktrace_sfr}
  Let $\lts = (\states, \init, \transitions)$ be an LTS and let $\normalizesfr(\lts) = (\states', \init', \transitions')$.
  For all sequences $\weaktrace \in \actions^*$ it holds that $\weaktrace \notin \weaktraces(\lts)$ if and only if $\init' \tracetransition{\weaktrace}' \emptyset$.
\end{restatable}

The structure explored by the refinement checking procedure of~\cite{Roscoe94,Roscoe10} for two LTSs $\lts_1$ and $\lts_2$ is the product $\normalizesfr(\lts_1) \product \lts_2$ in case of trace refinement and stable failures refinement.
For these structures the related witnesses, where the reachability of such a witness indicates non-refinement, are then as follows:

\begin{defi}[Witness]
  Let $\lts_1$ and $\lts_2$ be LTSs.
  A state $(U, s)$ of product $\normalizesfr(\lts_1) \product \lts_2$:
  \begin{itemize}
    \item is called a \emph{TR-witness} if and only if $U = \emptyset$.
    
    \item is called an \emph{SF-witness} if and only if at least one of the following conditions hold:
    \begin{itemize}
      \item $U = \emptyset$.
      
      \item $\stable(s)$ and $\refusals(s) \not\subseteq \refusals(\tostates{U}_{\lts_1})$.
    \end{itemize}  
  \end{itemize}
\end{defi}

\noindent 
We illustrate the notion of a witness, and in particular the relation between the reachability of a witness and the (violation of) the corresponding refinement relation by means of a small example.

\begin{exa}
  Consider the specification $s_0$ and the two implementations $t_0$ and $u_0$ as presented in Example~\ref{example:refinement} again. 
  In the figure below the (reachable part of the) normal form LTS of $s_0$ is depicted on the left, the product $\normalizesfr(s_0) \product t_0$ is shown in the middle and the product $\normalizesfr(s_0) \product u_0$ is shown on the right.
  \begin{center}
  \scriptsize
  \begin{tikzpicture}[->,node distance=45pt]
  	\tikzstyle{vertex}=[draw=none,minimum size=17pt,inner sep=2pt]

  	\node[vertex] (A1) {$\{s_0\}$};
  	\node[vertex] (A1init) [left=10pt of A1] {};
  	\node[vertex] (A1ten) [below of=A1, right  of=A1] {$\{s_3\}$};
  	\node[vertex] (A1right) [above of=A1ten, right  of=A1ten] {$\{s_1,s_2,s_5\}$};
    \node[vertex] (Aempty) [right of=A1] {$\emptyset$};

  	\draw (A1init) edge (A1);
  	\draw (A1) edge[bend left=90,looseness=1.5] node[above] {$\request$} (A1right);
  	\draw (A1right) edge[bend left] node[below right] {$10$} (A1ten);
  	\draw (A1right) edge[bend right=60,looseness=1.35] node[above] {$20$} (A1);
  	\draw (A1ten) edge[bend left] node[below left] {$10$} (A1);
    \draw (A1) edge node[below] {$10,20$} (Aempty);
    \draw (A1right) edge node[below] {$\request$} (Aempty);
    \draw (A1ten) edge node[right] {$20,\request$} (Aempty);
    \draw (Aempty) edge[loop above] node[above] {$10,20,\request$} (Aempty);
    
    \node[vertex] (D1) [above right=1cm and 5cm of A1] {$(\{s_0\}, t_0)$};
    \node[vertex] (D1init) [left=10pt of D1] {};
    \node[vertex] (D1right) [below=1cm of D1] {$(\{s_1,s_2,s_5\}, t_1)$};
    \node[vertex] (D1twenty) [below=1cm of D1right] {$(\{s_0\}, t_2)$};

    \draw (D1) edge node[right]{$\request$} (D1right);
    \draw (D1right) edge node[right]{$20$} (D1twenty);
    \draw (D1init) edge node[above] {} (D1);

    \node[vertex] (C1) [right=1.5cm of D1right] {$(\{s_0\}, u_0)$};
    \node[vertex] (C1init) [left=10pt of C1] {};
    \node[vertex] (C1right) [right=1cm of C1] {$(\{s_1,s_2,s_5\}, u_1)$};
    \node[vertex] (C1twenty) [below=1cm of C1right] {$(\{s_0\}, u_2)$};

    \draw (C1) edge node[above]{$\request$} (C1right);
    \draw (C1right) edge node[right]{$20$} (C1twenty);
    \draw (C1right) edge[loop above] node[above]{$\tau$} (C1right);
    \draw (C1twenty) edge node[below left]{$\tau$} (C1);
    \draw (C1init) edge node[above] {} (C1);
	\end{tikzpicture}
	\end{center}

  We observe that both $\normalizesfr(s_0) \product t_0$ and $\normalizesfr(s_0) \product u_0$ contain no TR-witnesses.
  In Example~\ref{example:refinement} we had already established that both $s_0 \refinedbytrace t_0$ and $s_0 \refinedbytrace u_0$.
  For the product $\normalizesfr(s_0) \product t_0$, the state $(\{s_0\}, t_2)$ is reachable and an SF-witness since $\stable(t_2)$ and $\{10,20,\request\} \in \refusals(t_2)$ but $\{10,20,\request\} \notin \refusals(\{s_0\})$.
  Intuitively, the product encodes that $\request\,20$ is a weak trace in both LTSs $\normalizesfr(s_0)$ and $t_2$ that reaches $\{s_0\}$ and $t_2$ respectively.
  Similarly, the normal form relates this weak trace to the reachability of $s_0$ from $s_0$.  
  Therefore, this witness indicates that $(\request\,20,\{10,20,\request\})$ is a failure of $t_0$, but not a failure of $s_0$, which establishes a violation of stable failures refinement.
  Finally, we can observe that $\normalizesfr(s_0) \product u_0$ contains no SF-witnesses.
  \qed
\end{exa}

\noindent The following lemmas formalise that trace refinement can be decided by checking reachability of a TR-witness in the product $\normalizesfr(\lts_1) \product \lts_2$.
Note that this result, and the related result for stable failures refinement, was already established in the literature; for instance, 
the relation between an SF-witness and the corresponding refinement relation can be found in~\cite{Roscoe94}.
However, the definitions in that paper are not explicit and the proof of this correspondence is only sketched.
Therefore, we here provide detailed proofs of these results.

\begin{lem}\label{lemma:traces_if}
  Let $\lts_1 = (\states_1, \init_1, \transitions_1)$ and $\lts_2 = (\states_2, \init_2, \transitions_2)$ be two LTSs.
  If $\lts_1 \refinedbytrace \lts_2$ holds then no TR-witness is reachable in $\normalizesfr(\lts_1) \product \lts_2$.
\end{lem}

\begin{proof}
  Suppose that $\lts_1 \refinedbytrace \lts_2$ holds, which means that $\weaktraces(\lts_2) \subseteq \weaktraces(\lts_1)$.
  Now assume that there is a reachable TR-witness $(\emptyset, s)$ in $\normalizesfr(\lts_1) \product \lts_2$.
  We show that this leads to a contradiction. 
  As the pair $(\emptyset, s)$ is reachable there is a weak trace $\weaktrace \in \actions^*$ such that $\init \weaktransition{\weaktrace} (\emptyset, s)$ where $\init$ is the initial state of $\normalizesfr(\lts_1) \product \lts_2$.
  From Proposition~\ref{lemma:product_traces} it follows that $\init_2 \weaktransition{\weaktrace}_2 s$ and $\emptyset$ is reachable by following $\weaktrace$ in $\normalizesfr(\lts_1)$.
  Therefore, $\weaktrace \in \weaktraces(\lts_2)$ and from Lemma~\ref{lemma:not_weaktrace_sfr} it follows that $\weaktrace \notin \weaktraces(\lts_1)$.
  This contradicts our assumption that $\weaktraces(\lts_2) \subseteq \weaktraces(\lts_1)$. Hence, no TR-witness is reachable in $\normalizesfr(\lts_1) \product \lts_2$.
\end{proof}

\begin{lem}\label{lemma:traces_only_if}
  Let $\lts_1 = (\states_1, \init_1, \transitions_1)$ and $\lts_2 = (\states_2, \init_2, \transitions_2)$ be two LTSs.
  If no TR-witness is reachable in $\normalizesfr(\lts_1) \product \lts_2$ then $\lts_1 \refinedbytrace \lts_2$.
\end{lem}

\begin{proof}
  Suppose that no TR-witness is reachable in $\normalizesfr(\lts_1) \product \lts_2$.
  Again, we prove this by contradiction.
  Assume that $\lts_1 \not\refinedbytrace \lts_2$ holds.
  This means that $\weaktraces(\lts_2) \not\subseteq \weaktraces(\lts_1)$.
  Pick a weak trace $\weaktrace \in \weaktraces(\lts_2)$ such that $\weaktrace \notin \weaktraces(\lts_1)$.
  Then there is a state $s \in S_2$ for which $\init_2 \weaktransition{\weaktrace}_2 s$.
  By Lemma~\ref{lemma:not_weaktrace_sfr} it holds that $\weaktrace$ leads to the empty set in $\normalizesfr(\lts_1)$.
  By Proposition~\ref{lemma:product_traces} the pair $(\emptyset, s)$ is then a reachable TR-witness in $\normalizesfr(\lts_1) \product \lts_2$.
  Contradiction.   
\end{proof}

\begin{thm}\label{theorem:traces}
  Let $\lts_1 = (\states_1, \init_1, \transitions_1)$ and $\lts_2 = (\states_2, \init_2, \transitions_2)$ be two LTSs.
  Then $\lts_1 \refinedbytrace \lts_2$ holds if and only if no TR-witness is reachable in $\normalizesfr(\lts_1) \product \lts_2$.
\end{thm}
\begin{proof}
  This follows directly from Lemmas~\ref{lemma:traces_if} and~\ref{lemma:traces_only_if}.
\end{proof}

\noindent Next, we formalise the relation between stable failures refinement and the reachability of an SF-witness.
In the proofs for the next two lemmas, we exploit Theorem~\ref{theorem:traces} and the fact that stable failures refinement is stronger than trace refinement.

\begin{lem}\label{lemma:stable_failures_if}
  Let $\lts_1 = (\states_1, \init_1, \transitions_1)$ and $\lts_2 = (\states_2, \init_2, \transitions_2)$ be two LTSs.
  If $\lts_1 \refinedbysfr \lts_2$ holds then no SF-witness is reachable in $\normalizesfr(\lts_1) \product \lts_2$.
\end{lem}
\begin{proof}
  Suppose that $\lts_1 \refinedbysfr \lts_2$ holds.
  Therefore, both $\failures(\lts_2) \subseteq \failures(\lts_1)$ and $\weaktraces(\lts_2) \subseteq \weaktraces(\lts_1)$.
  Now assume that there is a reachable SF-witness $(U, s)$ in $\normalizesfr(\lts_1) \product \lts_2$. 
  We show that this leads to a contradiction.
  For $(U, s)$ to be an SF-witness it holds that $U = \emptyset$ or both $\stable(s)$ and $\refusals(s) \not\subseteq \refusals(\tostates{U}_{\lts_1})$.
	However, since $\weaktraces(\lts_2) \subseteq \weaktraces(\lts_1)$ it follows that $\lts_1 \refinedbytrace \lts_2$ and, hence, by Theorem~\ref{theorem:traces}, no TR-witness is reachable in $\normalizesfr(\lts_1) \product \lts_2$.
  Consequently, $U \neq \emptyset$, and therefore it must be the case that
  $\stable(s)$ and $\refusals(s) \not\subseteq \refusals(\tostates{U}_{\lts_1})$ hold.
  
  As the pair $(U, s)$ is reachable there is a weak trace $\weaktrace \in \actions^*$ such that $\init \weaktransition{\weaktrace} (U, s)$ where $\init$ is the initial state of $\normalizesfr(\lts_1) \product \lts_2$.
  From Proposition~\ref{lemma:product_traces} it follows that $\init_2 \weaktransition{\weaktrace}_2 s$ and $U$ is reachable by following $\weaktrace$ in $\normalizesfr(\lts_1)$.    
	Since $\stable(s)$ and $\init_2 \weaktransition{\weaktrace}_2 s$, it follows that there must be a failure $(\weaktrace, X) \in \failures(\lts_2)$ where $s$ can 
  stably refuse $X \in \refusals(s)$, but $X \notin \refusals(\tostates{U}_{\lts_1})$.
  Let $(\weaktrace,X)$ be such.
  By Lemmas~\ref{lemma:deterministic} and~\ref{lemma:subset_existence_sfr} it follows for all states $t \in \states_1$ where $\init_1 \weaktransition{\weaktrace}_1 t$ that $t \in \tostates{U}_{\lts_1}$.
  For each stable $t$ it holds that $X \notin \refusals(t)$, because $X \notin \refusals(\tostates{U}_{\lts_1})$.  
  Therefore, we conclude that $(\weaktrace, X) \notin \failures(\lts_1)$, which contradicts $\failures(\lts_2) \subseteq 
  \failures(\lts_1)$.
  We conclude that the state $(U,s)$ cannot be an SF-witness.
\end{proof}

\begin{lem}\label{lemma:stable_failures_only_if}
  Let $\lts_1 = (\states_1, \init_1, \transitions_1)$ and $\lts_2 = (\states_2, \init_2, \transitions_2)$ be two LTSs.
  If no SF-witness is reachable in $\normalizesfr(\lts_1) \product \lts_2$ then $\lts_1 \refinedbysfr \lts_2$ holds.
\end{lem}
\begin{proof}
  Suppose that no SF-witness is reachable in $\normalizesfr(\lts_1) \product \lts_2$.
  Towards a contradiction, assume that $\lts_1 \not\refinedbysfr \lts_2$.
  By definition of the stable failures refinement this means that 
  $\failures(\lts_2) \not\subseteq \failures(\lts_1)$ or $\weaktraces(\lts_2) 
  \not\subseteq \weaktraces(\lts_1)$.
	If $\weaktraces(\lts_2) \not\subseteq \weaktraces(\lts_1)$ then $\lts_1 \not\refinedbytrace \lts_2$ and so, by Theorem~\ref{theorem:traces}, there must be a reachable TR-witness $(\emptyset, s)$, and, therefore, also a reachable SF-witness. Contradiction.
  
  Therefore $\failures(\lts_2) \not\subseteq \failures(\lts_1)$ and $\weaktraces(\lts_2) \subseteq \weaktraces(\lts_1)$.
  Pick a failure $(\weaktrace, X) \in \failures(\lts_2)$ such that 
  $(\weaktrace, X) \notin \failures(\lts_1)$.
  Since $(\weaktrace, X) \in \failures(\lts_2)$, there is a stable state $s \in \states_2$ such that $\init_2 \weaktransition{\weaktrace}_2 s$ and $X \in \refusals(s)$.
	Since $(\weaktrace,X) \in \failures(\lts_2)$, also $\weaktrace \in \weaktraces(\lts_1)$, and therefore $\weaktrace \in \weaktraces(\lts_1)$.
	By Lemmas~\ref{lemma:deterministic} and \ref{lemma:subset_existence_sfr} 
  weak trace $\weaktrace$ leads to a unique state $U$ in 
  $\normalizesfr(\lts_1)$ such that for all states $t \in \states_1$ with 
  $\init_1 \weaktransition{\weaktrace}_1 t$ it holds that $t \in 
  \tostates{U}_{\lts_1}$.  
  For each stable $t$ it holds that $X \notin \refusals(t)$, because $(\weaktrace, X) \notin \failures(\lts_1)$.
  Therefore, by Lemma~\ref{lemma:subset_sfr} it follows that $X \notin \refusals(\tostates{U}_{\lts_1})$.
  Hence, $\refusals(s) \not\subseteq \refusals(\tostates{U}_{\lts_1})$.  
  By Proposition~\ref{lemma:product_traces} the pair $(U, s)$ is reachable and 
  it is an SF-witness by definition.
  Contradiction.
\end{proof}

\begin{thm}\label{theorem:stable_failures}
  Let $\lts_1 = (\states_1, \init_1, \transitions_1)$ and $\lts_2 = (\states_2, \init_2, \transitions_2)$ be two LTSs.
  Then $\lts_1 \refinedbysfr \lts_2$ holds if and only if no SF-witness is reachable in $\normalizesfr(\lts_1) \product \lts_2$.
\end{thm}
\begin{proof}
  This follows directly from Lemmas~\ref{lemma:stable_failures_if} and~\ref{lemma:stable_failures_only_if}.
\end{proof}

\noindent
One may be inclined to believe that the normal form LTS can also be used to reduce the failures-divergences refinement decision problem to a reachability problem.
This is, however, not the case as the following example illustrates.
\begin{exa}\label{example:normalizefdr}
  Reconsider the LTSs $t_0$ and $u_0$ of Example~\ref{example:refinement}.
  Note that $t_0$ is a correct failures-divergences refinement of $u_0$, \ie, $u_0 \refinedbyfdr t_0$.
  The divergences $\request\,\weaktrace$, for $\weaktrace \in \actions^*$, of $u_0$ result in specification $u_0$ permitting all these sequences.
  \begin{center}
  \begin{tikzpicture}[->,node distance=45pt]
    \tikzstyle{vertex}=[draw=none,minimum size=17pt,inner sep=2pt]
    \scriptsize
  
    \node[vertex] (C1) {$\{u_0\}$};
    \node[vertex] (C1init) [left=10pt of C1] {};
    \node[vertex] (C1right) [right of=C1] {$\{u_1\}$};
    \node[vertex] (C1rightright) [right=45pt of C1right] {$\{u_0,u_2\}$};
    \node[vertex] (Cempty) [below of=C1right] {$\emptyset$};

    \draw (C1init) edge (C1);
    \draw (C1) edge node[above]{$\request$} (C1right);
    \draw (C1right) edge[bend left] node[above]{$20$} (C1rightright);
    \draw (C1rightright) edge[bend left] node[below]{$\request$} (C1right);
    \draw (C1) edge[bend right] node[below left] {$10,20$} (Cempty);
    \draw (C1right) edge node[left] {$10,\request$} (Cempty);
    \draw (C1rightright) edge[bend left] node[right] {$~10,20$} (Cempty);
    \draw (Cempty) edge[loop below] node[below] {$10,20,\request$} (Cempty);
    
    \node[vertex] (D1) [right=6cm of C1] {$(\{u_0\}, t_0)$};
    \node[vertex] (D1init) [left=10pt of D1] {};
    \node[vertex] (D2) [right= of D1] {$(\{u_1\}, t_1)$};
    \node[vertex] (D3) [below= of D2] {$(\{u_0, u_2\}, t_2)$};
    
    \draw (D1init) edge (D1);
    \draw (D1) edge node[above]{$\request$} (D2);
    \draw (D2) edge node[right]{$20$} (D3);
  \end{tikzpicture}
  \end{center}
  Consider the normal form $\normalizesfr(u_0)$ shown above on the left.
  The pair $(\{u_0,u_2\},t_2)$ in the product $\normalizesfr(u_0) \product t_0$ shown on the right, is thus problematic for the analysis of failures-divergence refinement, as the reachability of this pair might incorrectly indicate a violation of $u_0 \refinedbyfdr t_0$.
  In turn, this suggests that in the reachability analysis of the product, states beyond those reached via a trace constituting a divergence should not be considered candidate witnesses.
  \qed
\end{exa}
\noindent Our solution is to modify the construction of the normal form LTS for failures-divergences refinement as follows.

\begin{defi}[Failures-divergences normal form]\label{def:normalise}
  Let $\lts = (\states, \init, \transitions)$ be an LTS.  
  The failures-divergences normal form of $\lts$ is the LTS $\normalize(\lts) = (\states', \init', \transitions')$, where $\states' = \powerset(\states)$, $\init' = \{s \in \states \mid \init \weaktransition{\emptytrace} s \}$ and $\transitions'$ is defined as $U \transition{a}' V$ if and only if $\neg (\exists s \in U : \diverges{s})$ and $V = \{t \in \states \mid \exists s \in U : s \weaktransition{a} t\}$ for all sets of states $U, V \subseteq \states$ and actions $a \in \actions$.
\end{defi}

\noindent Notice that $\normalize(\lts)$ yields a subgraph of $\normalizesfr(\lts)$.
As a result, several properties that we established for $\normalizesfr$ carry over to $\normalize$.
For instance, $\normalize$ yields LTSs that are deterministic and concrete.
However, contrary to LTSs obtained via $\normalizesfr$, LTSs obtained via $\normalize$ are not guaranteed to be universal.
In particular, a weak trace $\weaktrace \in \weaktraces(\lts)$ is not guaranteed to be preserved in $\normalize(\lts)$ if it is a divergence. 
Consequently, Lemma~\ref{lemma:not_weaktrace_sfr}, which is essential for Theorems~\ref{theorem:traces} and~\ref{theorem:stable_failures}, no longer holds in its full generality.
We show, however, that for failures-divergence refinement a slightly different relation between an LTS and its normal form is sufficient for establishing a theorem that is similar in spirit to the aforementioned theorems.

\begin{lem}\label{lemma:subset_fdr}
  Let $\lts = (\states, \init, \transitions)$ be an LTS and let $\normalize(\lts) = (\states', \init', \transitions')$.
  For all sequences $\weaktrace \in \actions^*$ and states $U \in \states'$ such that $\init' \tracetransition{\weaktrace}' U$ it holds that $\init \weaktransition{\weaktrace} s$ for all $s \in \tostates{U}_\lts$.
\end{lem}  

\begin{proof}
  Along the same lines as the proof of Lemma~\ref{lemma:subset_sfr}.
\end{proof}

\noindent We mentioned that divergences are not necessarily preserved (as traces) by the normalisation.
In fact, we can be more specific: only \emph{minimal} divergences are preserved in the normal form LTS.
 The \emph{minimal} divergences of a state $s \in \states$, denoted by $\divergencesbot(s)$, is the largest subset of $\divergences(s)$ containing all $\weaktrace \in \divergences(s)$ for which there is no strict prefix of $\weaktrace$ in $\divergences(s)$.
For an LTS $\lts = (\states, \init, \transitions)$ we define $\divergencesbot(\lts)$ to be $\divergencesbot(\init)$.

\begin{restatable}{lem}{lemmasubsetexistencefdr}\label{lemma:subset_existence_fdr}
  Let $\lts = (\states, \init, \transitions)$ be an LTS and let $\normalize(\lts) = (\states', \init', \transitions')$.
  For all sequences $\weaktrace \in \actions^*$ such that either $\weaktrace \notin \divergences(\lts)$ or $\weaktrace \in \divergencesbot(\lts)$ and for all states $s \in \states$ such that $\init \weaktransition{\weaktrace} s$ there is a state $U \in \states'$ such that $s \in \tostates{U}_\lts$ and $\init' \tracetransition{\weaktrace}' U$.
\end{restatable}

\begin{restatable}{lem}{lemmanotdivergencefdr}\label{lemma:not_divergence_fdr}
  Let $\lts = (\states, \init, \transitions)$ be an LTS and let $\normalize(\lts) = (\states', \init', \transitions')$.
  For all sequences $\weaktrace \in \actions^*$ and states $U \in \states'$ it holds that if $\init' \tracetransition{\weaktrace}' U$ and not $\diverges{\tostates{U}_\lts}$ then $\weaktrace \notin \divergences(\lts)$.
\end{restatable}

\begin{restatable}{lem}{lemmanotweaktracefdr}\label{lemma:not_weaktrace_fdr}
  Let $\lts = (\states, \init, \transitions)$ be an LTS and let $\normalize(\lts) = (\states', \init', \transitions')$.
  For all sequences $\weaktrace \in \actions^*$ it holds that $\weaktrace \notin (\divergences(\lts) \cup \weaktraces(\lts))$ if and only if $\init' \tracetransition{\weaktrace}' \emptyset$.
\end{restatable}

\noindent For failures-divergences refinement the state space of $\normalize(\lts_1) \product \lts_2$ is explored for a witness, where reachability of such a witness also indicates non-refinement.
This witness is defined as follows:

\begin{defi}[Failures-divergences witness]
  Let $\lts_1$ and $\lts_2$ be two LTSs.
  A state $(U, s)$ of the product $\normalize(\lts_1) \product \lts_2$ is called an \emph{FD-witness} if and only if $\diverges{\tostates{U}_{\lts_1}}$ does not hold and at least one of the following conditions hold:
  \begin{itemize}
    \item $U = \emptyset$.
    
    \item $\stable(s)$ and $\refusals(s) \not\subseteq \refusals(\tostates{U}_{\lts_1})$.
    
    \item $\diverges{s}$.
  \end{itemize}  
\end{defi}

\noindent We next formalise the correspondence between failures-divergences refinement and the reachability of an FDR-witness.
In the proof of the following theorem we cannot easily use Theorem~\ref{theorem:traces} because $\refinedbyfdr$ is incomparable with both $\refinedbytrace$ and $\refinedbysfr$.

\begin{lem}\label{lemma:failures_if}
  Let $\lts_1 = (\states_1, \init_1, \transitions_1)$ and $\lts_2 = (\states_2, \init_2, \transitions_2)$ be two LTSs.
  If $\lts_1 \refinedbyfdr \lts_2$ holds then no FD-witness is reachable in $\normalize(\lts_1) \product \lts_2$.
\end{lem}
\begin{proof}
  Assume that $\lts_1 \refinedbyfdr \lts_2$. We then have that $\failuresbot(\lts_2) \subseteq \failuresbot(\lts_1)$ and $\divergences(\lts_2) \subseteq \divergences(\lts_1)$.
  Towards a contradiction, assume there is an FD-witness $(U, s)$ in $\reachable(\normalize(\lts_1) \product \lts_2)$.
	Let $\init$ be the initial state of $\normalize(\lts_1) \product \lts_2$, and let $\weaktrace \in \weaktraces(\init)$ be such that $\init \weaktransition{\weaktrace} (U, s)$.
  From the assumption that $(U, s)$ is an FD-witness it follows that not $\diverges{\tostates{U}_{\lts_1}}$ and from Lemma~\ref{lemma:not_divergence_fdr} it follows that $\weaktrace \notin \divergences(\lts_1)$.
  By Proposition~\ref{lemma:product_traces} it holds that $\init_2 \weaktransition{\weaktrace}_2 s$ and $U$ is reachable by following $\weaktrace$ in $\normalize(\lts_1)$.
	Moreover, for $(U,s)$ to be an FD-witness, at least one of the following must also hold: $U = \emptyset$, $\stable(s)$ and $\refusals(s) \not\subseteq \refusals(\tostates{U}_{\lts_1})$, or $\diverges{s}$.
  We therefore distinguish these three cases:

  \begin{itemize}
  \item Case $U = \emptyset$.
    We can assume that $\diverges{s}$ does not hold, as this is handled by another case.
    Then it follows that there is a state $t \in \states_2$ such that $\init_2 \weaktransition{\weaktrace}_2 s \weaktransition{\emptytrace}_2 t$ and $\stable(t)$. Let $t$ be such.
    Consequently, $(\weaktrace,X) \in \failuresbot(\lts_2)$ for some $X \in \refusals(t)$.
    By Lemma~\ref{lemma:not_weaktrace_fdr} it holds that the weak trace $\weaktrace$ reaching the empty set in $\normalize(\lts_1)$ is not a weak trace of $\lts_1$.
    Together with $\weaktrace \notin \divergences(\lts_1)$ it follows, for all 
    possible refusal sets $Y \subseteq \actions$, that 
    $(\weaktrace, Y) \notin \failuresbot(\lts_1)$, and so, in particular, $(\weaktrace, X) \in \failuresbot(\lts_1)$ which contradicts our assumption that $\failuresbot(\lts_2) \subseteq 
    \failuresbot(\lts_1)$.
    
  \item Case $\stable(s)$ and $\refusals(s) \not\subseteq \refusals(\tostates{U}_{\lts_1})$.
    From $\stable(s)$ and the reachability of state $s$ it follows that $\failuresbot(\lts_2)$ is not empty.
    Pick a failure $(\weaktrace, X) \in \failuresbot(\lts_2)$ where $s$ can stably refuse $X \in \refusals(s)$, but $X \notin \refusals(\tostates{U}_{\lts_1})$.
    By Lemmas~\ref{lemma:deterministic} and \ref{lemma:subset_existence_fdr} it follows for all states $t \in \states_1$ where $\init_1 \weaktransition{\weaktrace}_1 t$ that $t \in \tostates{U}_{\lts_1}$.
    For each stable $t$ it holds that $X \notin \refusals(t)$, because $X \notin \refusals(\tostates{U}_{\lts_1})$.
    Due to the previous case, we may assume that $U \neq \emptyset$.
    Then from $\weaktrace \notin \divergences(\lts_1)$ and $U \neq \emptyset$ it follows that $(\weaktrace, X) \notin \failuresbot(\lts_1)$, which leads to a contradiction with the assumption that $\failuresbot(\lts_2) \subseteq \failuresbot(\lts_1)$.
  
  \item Case $\diverges{s}$.
	  Since $\init_2 \weaktransition{\weaktrace}_2 s$ and $\diverges{s}$, also $\weaktrace \in \divergences(\lts_2)$. 
    However, by $\weaktrace \notin \divergences(\lts_1)$ this contradicts the assumption that $\divergences(\lts_2) \subseteq \divergences(\lts_1)$.\qedhere
  \end{itemize}
\end{proof}

\begin{lem}\label{lemma:failures_only_if}
  Let $\lts_1 = (\states_1, \init_1, \transitions_1)$ and $\lts_2 = (\states_2, \init_2, \transitions_2)$ be two LTSs.
  If no FD-witness is reachable in $\normalize(\lts_1) \product \lts_2$ then $\lts_1 \refinedbyfdr \lts_2$ holds.
\end{lem}
\begin{proof}
  Assume that no FD-witness is reachable in $\normalize(\lts_1) \product \lts_2$.
  Again, we prove this by contradiction.
  Assume that $\lts_1 \not\refinedbyfdr \lts_2$.
  By definition of failures-divergences refinement this means that $\failuresbot(\lts_2) \not\subseteq \failuresbot(\lts_1)$ or $\divergences(\lts_2) \not\subseteq \divergences(\lts_1)$.
  Hence, there are two cases to consider:
  
  \begin{itemize}
  \item Case $\divergences(\lts_2) \not\subseteq \divergences(\lts_1)$.
    Pick a diverging weak trace $\weaktrace \in \divergences(\lts_2)$ such that $\weaktrace \notin \divergences(\lts_1)$.
    In this case there is a prefix of $\weaktrace$, which we call $\trace$, that leads to a diverging state $s \in \states_2$ such that $\init_2 \weaktransition{\trace}_2 s$.
    However, by the assumption that $\weaktrace \notin \divergences(\lts_1)$ we know that all states $t \in \states_1$ reached by following $\trace$ are not diverging.
    By Lemma~\ref{lemma:subset_fdr} we know that all $t \in U$ can be reached by following $\trace$.
    Therefore state pair $(U, s)$ is an FD-witness, because $\diverges{s}$ but not $\diverges{\tostates{U}_{\lts_1}}$.
    Contradiction.

  \item Case $\failuresbot(\lts_2) \not\subseteq \failuresbot(\lts_1)$.
    By the previous case, we may, moreover, assume that $\divergences(\lts_2) \subseteq \divergences(\lts_1)$.
    Pick any failure $(\weaktrace, X) \in \failuresbot(\lts_2)$ such that $(\weaktrace,X) \notin \failuresbot(\lts_1)$.
    Observe that $\weaktrace \notin \divergences(\lts_1)$ (and as such $\weaktrace \notin \divergences(\lts_2))$ as otherwise no such $(\weaktrace, X)$ exists by definition of $\failuresbot(\lts_1)$.
    Since $(\weaktrace, X) \in \failuresbot(\lts_2)$, there is a stable state $s \in \states_2$ such that $\init_2 \weaktransition{\weaktrace}_2 s$ and $X \in \refusals(s)$.
    We distinguish whether weak trace $\weaktrace$ is among the weak traces of $\lts_1$ or not:
  
    \begin{itemize}
    \item Case $\weaktrace \notin \weaktraces(\lts_1)$.
      By Lemma \ref{lemma:not_weaktrace_fdr} this means that $\weaktrace$ is a trace leading to the empty set in $\normalize(\lts_1)$.
		    By Proposition~\ref{lemma:product_traces}, a pair $(\emptyset, s)$ is then reachable in $\normalize(\lts_1) \product \lts_2$. But then that pair is a reachable FD-witness.
      Contradiction.
    
    \item Case $\weaktrace \in \weaktraces(\lts_1)$.
      Recall that $(\weaktrace, X) \notin \failuresbot(\lts_1)$ by assumption.
      By Lemmas~\ref{lemma:deterministic} and \ref{lemma:subset_existence_fdr} there is a unique state $V$ of $\normalize(\lts_1)$ reachable via weak trace $\weaktrace$ such that for all $t \in \states_1$ where $\init_1 \weaktransition{\weaktrace}_1 t$, it holds that $t \in \tostates{V}_{\lts_1}$.
      For each stable state $t$ it holds that $X \notin \refusals(t)$, because $(\weaktrace, X) \in \failuresbot(\lts_1)$.
      Therefore, by Lemma~\ref{lemma:subset_fdr} it follows that $X \notin \refusals(\tostates{V}_{\lts_1})$.
      Therefore $\refusals(s) \not\subseteq \refusals(\tostates{V}_{\lts_1})$.  
      By Proposition~\ref{lemma:product_traces} the pair $(V, s)$ is reachable and it is an FD-witness by definition.
      Contradiction.\qedhere
    \end{itemize}
    
  \end{itemize}
\end{proof}

\begin{thm}\label{theorem:failures}
  Let $\lts_1 = (\states_1, \init_1, \transitions_1)$ and $\lts_2 = (\states_2, \init_2, \transitions_2)$ be two LTSs.
  Then $\lts_1 \refinedbyfdr \lts_2$ holds if and only if no FD-witness is reachable in $\normalize(\lts_1) \product \lts_2$.
\end{thm}

\begin{proof}
  This follows directly from Lemmas~\ref{lemma:failures_if} and~\ref{lemma:failures_only_if}.
\end{proof}

\begin{exa}
  Consider the LTSs $t_0$ and $u_0$ of Example~\ref{example:refinement} once more.
  As we noted in Example~\ref{example:normalizefdr}, $t_0$ is a correct failures-divergences refinement of $u_0$.
  In the figure below the normal form $\normalize(u_0)$ is shown on the left and the product $\normalize(u_0) \product t_0$ is shown on the right.
  \begin{center}
  \begin{tikzpicture}[->,node distance=45pt]
    \tikzstyle{vertex}=[draw=none,minimum size=17pt,inner sep=2pt]
    \scriptsize
  
    \node[vertex] (C1) {$\{u_0\}$};
    \node[vertex] (C1init) [left=10pt of C1] {};
    \node[vertex] (C1right) [right of=C1] {$\{u_1\}$};
    \node[vertex] (Cempty) [below of=C1right] {$\emptyset$};

    \draw (C1init) edge (C1);
    \draw (C1) edge node[above]{$\request$} (C1right);
    \draw (C1) edge node[below left] {$10,20$} (Cempty);
    
    \node[vertex] (D1) [right=6cm of C1] {$(\{u_0\}, t_0)$};
    \node[vertex] (D1init) [left=10pt of D1] {};
    \node[vertex] (D2) [right= of D1] {$(\{u_1\}, t_1)$};
    
    \draw (D1init) edge (D1);
    \draw (D1) edge node[above]{$\request$} (D2);
    \draw (Cempty) edge[loop below] node[below] {$10,20,\request$} (Cempty);
  \end{tikzpicture}
  \end{center}
  Note that the pair $(\{u_0,u_2\},t_2)$, which was reachable in the product
  $\normalizesfr(u_0) \product t_0$, is no longer reachable in the product
  $\normalize(u_0) \product t_0$.
  In fact, no FD-witness is reachable in the latter product, thus confirming that indeed
  $u_0 \refinedbyfdr t_0$.\qed
\end{exa}

\section{Antichain Algorithms for Refinement Checking}\label{section:original-algorithm}

Notice that Theorems~\ref{theorem:traces},~\ref{theorem:stable_failures} and~\ref{theorem:failures} provide the basis for straightforward algorithms for deciding trace refinement, stable failures refinement and failures-divergences refinement: one can explore the product of the normalised specification and the impementation, looking for a witness \emph{on-the-fly}.
In these algorithms, the normalisation of the specification LTS dominates the theoretical worst-case run time complexity of the algorithms. While refinement checking itself is a PSPACE-hard problem, in practice, the problem can often be solved quite effectively.
Nevertheless, as observed in~\cite{WangS0LDWL12}, antichains provide room for improvement by  potentially reducing the number of states of the normal form LTS of the specification that must be checked.

An antichain is a set $\mathcal{A} \subseteq X$ of a partially ordered set $(X,\leq)$ in which all distinct $x,y \in \mathcal{A}$ are incomparable: neither $x \leq y$ nor $y \leq x$.
Given a partially ordered set $(X,\leq)$ and an antichain $\mathcal{A}$, the membership test, denoted by $\chainin$, checks whether $\mathcal{A}$ `contains' an element $x$; that is, $x \chainin \mathcal{A}$ holds true if and only if there is some $y \in \mathcal{A}$ such that $y \leq x$.  
We write $Y \chainsetin \mathcal{A}$ iff $y \chainin \mathcal{A}$ for all $y \in Y$.  Antichain $\mathcal{A}$ can be extended by inserting an element $x \in X$, denoted $\mathcal{A} \chaininsert x$, which is defined as the set $\{ y \mid y = x \vee (y \in \mathcal{A} \wedge x \not\leq y)\}$. 
Note that this operation only yields an antichain whenever $x \not\chainin \mathcal{A}$.

As~\cite{WangS0LDWL12,AbdullaCHMV10} suggest, the state space of the product $(\states, \init ,\transitions)$ between a normal form of LTS $\lts_1$ and the LTS $\lts_2$ induces a partially ordered set as follows. 
For $(U,s),(V,t) \in \states$, define $(U,s) \statesubset (V,t)$ iff $s = t$ and $\tostates{U}_{\lts_1} \subseteq \tostates{V}_{\lts_1}$. 
Then the set $(S,\statesubset)$ is a partially ordered set.
The fundamental property underlying the reason why an antichain approach to refinement checking works is expressed by the following claim (which we repeat as Proposition~\ref{lemma:subset_trace}, and prove in Section~\ref{section:improved-algorithm}), stating that the traces of any state $(V,s)$ in the product can be executed from all states smaller than $(V,s)$. 
Notice that this property relies on the fact that the empty set is included as a state in the normal form LTS.

\begin{claim}
  For all states $(U, s), (V, s)$ of $\normalize(\lts_1) \product \lts_2$ satisfying $(U, s) \statesubset (V, s)$ and for every sequence $\trace \in \actionstau^*$ such that $(V, s) \tracetransition{\trace} (V', t)$ there is a state $(U',t)$ such that  $(U, s) \tracetransition{\trace} (U', t)$ and $(U', t) \statesubset (V', t)$.
\end{claim}

\noindent The main idea of the antichain algorithms is now as follows: the set of states of the product that have been explored are recorded in an antichain rather than a set. 
Whenever a new state of the product is found that is already included in the antichain (w.r.t.~the membership test $\chainin$), further exploration of that state is unnecessary, thereby pruning the state space of the product. 
While the proposition stated above suggests this is sound for trace refinement, it is not immediate that doing so is also sound for refusals and divergences.

Based on the above informal reasoning,~\cite{WangS0LDWL12} presents antichain algorithms that intend to check for trace refinement, stable failures refinement and failures-divergences refinement.
Before we discuss these algorithms in more detail, see Algorithms~\ref{alg:original_tr}-\ref{alg:original_fdr}, we here present their pseudocode for the sake of completeness; in the remainder of this paper, we refer to these as the \emph{original} algorithms.

\begin{rem}\label{remark:refusals}
  For the implementation of refusals in Algorithms~\ref{alg:original_sfr} and~\ref{alg:original_fdr} we followed the definition of $\refusals$ provided in~\cite{WangS0LDWL12}. 
  This definition differs subtly from Definition~\ref{def:refusals}, by defining, for any (not necessarily stable) state $s$, $\refusals(s) = \{X \mid \exists s' \in \states :  (s  \weaktransition{\emptytrace} s' \land \stable(s') \land X \subseteq \actions \setminus \enabled(s')) \}$ and $\refusals(U) = \{X \mid \exists s \in U : X \in \refusals(s)\}$ for $U \subseteq S$.
\end{rem}
 
\begin{algorithm}[H]
\footnotesize
  \caption{Antichain-based trace refinement algorithm presented in~\cite{WangS0LDWL12}.
  The algorithm returns \emph{true} if and only if $\lts_1 = (\states_1, \init_1, \transitions_1)$ is refined by $\lts_2  = (\states_2, \init_2, \transitions_2)$ in trace semantics.}
  \label{alg:original_tr}
  \begin{algorithmic}[1]
    \Procedure{refines-trace}{$\lts_1, \lts_2$}
    \State {let $\working$ be a stack containing a pair $( \{s \in 
    \states_1 \mid \init_1 \weaktransition{\emptytrace}_1 s\}, \init_2)$ }
    \State {let $\antichain \gets \emptyset$}
    \While {$\working \neq \emptyset$}
      \State {pop ($\spec, \impl$) from $\working$}
      \State {$\antichain \gets \antichain \chaininsert (\spec, \impl)$}
      \For {$\impl \transition{a}_2 \impl'$}
        \If {$a = \tau$}
          \State {$\spec' \gets \spec$}
        \Else
          \State {$\spec' \gets \{s' \in \states_1  \mid \exists s \in 
          \spec : s \weaktransition{a}_1 s'\}$}			
        \EndIf
        \If {$\spec' = \emptyset$}
          \Return false
        \EndIf
        \If {$(\spec', \impl') \chainnotin \antichain$}
          \State {push $(\spec', \impl')$ into $\working$}\label{alg:original_tr:push_working}
        \EndIf	
      \EndFor
    \EndWhile
    \Return true
    \EndProcedure
  \end{algorithmic}
\end{algorithm}

\begin{algorithm}[H]
\footnotesize
  \caption{Antichain-based stable failures refinement algorithm presented in~\cite{WangS0LDWL12}.
  The algorithm returns \emph{true} if and only if $\lts_1 = (\states_1, \init_1, \transitions_1)$ is refined by $\lts_2  = (\states_2, \init_2, \transitions_2)$ in stable failures semantics.}  \label{alg:original_sfr}
  \begin{algorithmic}[1]
    \Procedure{refines-stable-failures}{$\lts_1, \lts_2$}
    \State {let $\working$ be a stack containing a pair $( \{s \in 
    \states_1 \mid \init_1 \weaktransition{\emptytrace}_1 s\}, \init_2)$ }
    \State {let $\antichain \gets \emptyset$}
    \While {$\working \neq \emptyset$}
      \State {pop ($\spec, \impl$) from $\working$}
      \State {$\antichain \gets \antichain \chaininsert (\spec, \impl)$}
      \If { $\refusals(\impl) \not\subseteq 
      \refusals(\spec)$}\label{alg:original_sfr:refusals}
        \Return false
      \EndIf
      \For {$\impl \transition{a}_2 \impl'$}
        \If {$a = \tau$}
          \State {$\spec' \gets \spec$}
        \Else
          \State {$\spec' \gets \{s' \in \states_1  \mid \exists s \in 
          \spec : s \weaktransition{a}_1 s'\}$}			
        \EndIf
        \If {$\spec' = \emptyset$}
          \Return false
        \EndIf
        \If {$(\spec', \impl') \chainnotin \antichain$}
          \State {push $(\spec', \impl')$ into $\working$}
        \EndIf	
      \EndFor
    \EndWhile
    \Return true
    \EndProcedure
  \end{algorithmic}
\end{algorithm}

\begin{algorithm}[H]
  \footnotesize
  \caption{Antichain-based failures-divergences refinement algorithm presented in~\cite{WangS0LDWL12}.
  The algorithm is claimed to return \emph{true} iff $\lts_1 = (\states_1, \init_1, \transitions_1)$ is refined by $\lts_2  = (\states_2, \init_2, \transitions_2)$ in failures-divergences semantics.}
  \label{alg:original_fdr}
  \begin{algorithmic}[1]
    \Procedure{refines-failures-divergences}{$\lts_1, \lts_2$}
    \State {let $\working$ be a stack containing a pair $( \{s \in \states_1 \mid \init_1 \weaktransition{\emptytrace}_1 s\}, \init_2)$ }
    \State {let $\antichain \gets \emptyset$}
    \While {$\working \neq \emptyset$}
      \State {pop ($\spec, \impl$) from $\working$}
      \State {$\antichain \gets \antichain \chaininsert (\spec, \impl)$}
      \If {$\diverges{\impl}$}\label{alg:original_fdr:impl_diverges}
        \If {not $\diverges{\spec}$}\label{alg:original_fdr:spec_diverges}
          \Return false
        \EndIf
      \Else
      \If { $\refusals(\impl) \not\subseteq 
      \refusals(\spec)$}\label{alg:original_fdr:refusals}
        \Return false
      \EndIf
      \For {$\impl \transition{a}_2 \impl'$}
        \If {$a = \tau$}
          \State {$\spec' \gets \spec$}
        \Else
          \State {$\spec' \gets \{s' \in \states_1  \mid \exists s \in 
          \spec : s \weaktransition{a}_1 s'\}$}			
        \EndIf
        \If {$\spec' = \emptyset$}
          \Return false
        \EndIf
        \If {$(\spec', \impl') \chainnotin \antichain$}
          \State {push $(\spec', \impl')$ into $\working$}
        \EndIf	
      \EndFor
      \EndIf
    \EndWhile
    \Return true
    \EndProcedure
  \end{algorithmic}
\end{algorithm}

\noindent Let us stress that Algorithm~\ref{alg:original_tr} correctly decides trace refinement and Algorithm~\ref{alg:original_sfr} correctly decides stable failures refinement.
However, Algorithm~\ref{alg:original_fdr} fails to correctly decide failures-divergences
refinement. 
Moreover, it is interesting to note that Algorithms~\ref{alg:original_sfr} and~\ref{alg:original_fdr} fail to decide the \emph{non-standard} relations used in~\cite{WangS0LDWL12}, see also the discussion in Remark~\ref{remark:incorrect}. 
All three issues are illustrated by the example below.

\begin{exa}\label{ex:incorrect} Consider the four transition systems depicted below.
\begin{center}
\scriptsize
\begin{tikzpicture}[->]
  \tikzstyle{vertex}=[draw=none,minimum size=17pt,inner sep=0pt]
  
  \node[vertex] (A0) {$s_0$};
  \node[vertex] (A0init) [left=10pt of A0] {};
  \node[vertex] (A1) [right=2cm of A0] {$s_1$};
  \node[vertex] (A1init) [left=10pt of A1] {};
  \node[vertex] (A2) [right=2cm of A1] {$s_2$};
  \node[vertex] (A2') [right=.5cm of A2] {$t_2$};
  \node[vertex] (A2init) [left=10pt of A2] {};
  \node[vertex] (A3) [right=2.5cm of A2] {$s_3$};
  \node[vertex] (A3') [right=.5cm of A3] {$t_3$};
  \node[vertex] (A3init) [left=10pt of A3] {};
 
  \draw (A0init) edge (A0); 
  \draw (A1init) edge (A1); 
  \draw (A2init) edge (A2); 
  \draw (A3init) edge (A3); 
  \draw (A0) edge[loop above] node[above] {$\tau$} (A0);
  \draw (A0) edge[loop right] node[right] {$a$} (A0);
  \draw (A1) edge[loop right] node[right] {$b$} (A1);
  \draw (A2) edge[loop above] node[above] {$\tau$} (A2);
  \draw (A2) edge[bend left] node[above] {$a$} (A2');
  \draw (A2') edge[bend left] node[below] {$a$} (A2);
  \draw (A3) edge[loop above] node[above] {$\tau$} (A3);
  \draw (A3) edge node[above] {$a$} (A3');
\end{tikzpicture} 
\end{center}

\noindent
Let us first observe that Algorithm~\ref{alg:original_sfr} correctly decides that $s_1
\refinedbysfr s_0$ does not hold, which follows from a violation of
$\weaktraces(s_0) \subseteq \weaktraces(s_1)$.  Next,
observe that we have $s_0 \refinedbyfdr s_1$, since the divergence
of the root state $s_0$ implies chaotic behaviour of $s_0$ and,
hence, any system refines such a system. It is not hard to see, however, that 
Algorithm~\ref{alg:original_fdr} returns \emph{false}, wrongly concluding that $s_0 \not\refinedbyfdr s_1$.

With respect to the non-standard refinement relations defined in~\cite{WangS0LDWL12}, see also Remark~\ref{remark:incorrect},
we observe the following.
Since $s_0$ is not stable, we have $\failures(s_0) = \emptyset$ and hence
$\failures(s_0) \subseteq \failures(s_1)$. Consequently, stable failures
refinement as defined in~\cite{WangS0LDWL12} should hold, but as we already
concluded above, the algorithm returns \emph{false} when checking for $s_1 \refinedbysfr s_0$.
Next, observe that the algorithm returns \emph{true} when checking for $s_2 \refinedbyfdr s_3$.
The reason is that for the pair $(\{s_2\},s_3)$, it detects that state $s_3$ diverges and concludes that
since also the normal form state of the specification $\{s_2\}$ diverges, it
can terminate the iteration and return \emph{true}. This is a consequence of
splitting the divergence tests over two \textbf{if}-statements in lines~\ref{alg:original_fdr:impl_diverges} 
and~\ref{alg:original_fdr:spec_diverges}.
According to the failures-divergences refinement of~\cite{WangS0LDWL12}, however,
the algorithm should return \emph{false}, since $\failures(s_3) \subseteq
\failures(s_2)$ fails to hold: we have $(a, \{a\}) \in \failures(s_3)$ but
not $(a,\{a\}) \in \failures(s_2)$.\qed
\end{exa}

\noindent Notice that each algorithm explores the product between the normal form of a specification, and an implementation in a depth-first, on-the-fly manner.
While depth-first search is typically used for detecting divergences, \cite{Roscoe94} states a number of reasons for running a refinement check in a breadth-first manner.
Indeed, a compelling argument in favour of using a breadth-first search is conciseness of the counterexample in case of a non-refinement.

Each algorithm can be made to run in a breadth-first fashion simply by using a FIFO \emph{queue} rather than a stack as the data structure for $\working$. 
However, our implementations of these algorithms suffer from severely degraded performance. 
The performance degradation can be traced back to the following three additional problems in the original algorithms, which also are present (albeit less pronounced in practice) when utilising a depth-first exploration:

\begin{enumerate}
  \item The refusal check on line~\ref{alg:original_sfr:refusals} of Algorithm~\ref{alg:original_sfr} (and line~\ref{alg:original_fdr:refusals} of Algorithm~\ref{alg:original_fdr}) is also performed for unstable states, which, combined with the definition of $\refusals$ in~\cite{WangS0LDWL12} (see also Remark~\ref{remark:refusals}), results in a repeated, potentially expensive, search for stable states;
  
  \item In all three algorithms, duplicate pairs might be added to $\working$ since $\working$ is filled with all successors of $(\spec, \impl)$ that fail the $\antichain$ membership test, regardless of whether these pairs are already scheduled for exploration, \ie, included in $\working$, or not;
    
  \item In all three algorithms, contrary to the explicit claim in~\cite[Section 2.2]{WangS0LDWL12} the variable $\antichain$ is not guaranteed to be an antichain.  
\end{enumerate}
  
\noindent The first problem is readily seen to lead to undesirable overhead. 
The second and third problem are more subtle. 
We first illustrate the second problem on Algorithm~\ref{alg:original_tr}: the following example shows a case where the algorithm stores an excessive number of pairs in $\working$.
Note that the two other algorithms suffer from the same phenomenon.

\begin{exa}\label{example:badcase}
Consider the family of LTSs $\lts_n^k = (\states_n, \init_n, \transitions_n)$ with states $\states_n = \{s_1, \dots, s_n\}$, transitions $s_i \transition{a_j}_n s_{i-1}$ for all $1 \leq j \leq k$, $1 < i \leq n$ and $\init_n = s_n$; see also the transition system depicted below.
Note that each LTS that belongs to this family is completely deterministic and concrete.

\begin{center}
\scriptsize
  \begin{tikzpicture}[->]
  \tikzstyle{vertex}=[draw=none,minimum size=25pt,inner sep=0pt]
  
  \node[vertex,minimum size=15pt] (A0) {$s_n$};
  \node[vertex] (A0init) [left=10pt of A0] {};
  \node[vertex] (A1) [right=1cm of A0] {$s_{n-1}$};
  \node (A2) [right=1cm of A1] {};
  \node (A23) [right=.2cm of A2] {$\dots$};
  \node (A3) [right=.2cm of A23] {$s_2$};
  \node[vertex] (A4) [right=1cm of A3] {$s_1$};
  
  \draw (A0init) edge (A0); 
  \draw (A0) edge[bend right=30] node[below] {$a_k$} (A1);
  \draw (A0) edge[draw=none] node[yshift=2.5pt] {$\vdots$} (A1);
  \draw (A0) edge[bend left=30] node[above] {$a_1$} (A1);  
  
  \draw (A1) edge[dotted, bend right=30] node[below] {$a_k$} (A2);
  \draw (A1) edge[draw=none] node[yshift=2.5pt] {$\vdots$} (A2);
  \draw (A1) edge[dotted, bend left=30] node[above] {$a_1$} (A2); 
  
  \draw (A3) edge[bend right=30] node[below] {$a_k$} (A4); 
  \draw (A3) edge[draw=none] node[yshift=2.5pt] {$\vdots$} (A4); 
  \draw (A3) edge[bend left=30] node[above] {$a_1$} (A4);  
  \end{tikzpicture} 
\end{center}

\noindent Each labelled transition system in this class has $n$ states and $k \cdot (n - 1)$ transitions. 
Suppose one checks for trace refinement between an implementation and specification both of which are given by $\lts_n^k$; \ie, we test for $\lts_n^k \refinedbytrace \lts_n^k$.  

Using a depth-first search, Algorithm~\ref{alg:original_tr} will add the state reachable via a single step once for every action, because $(\{s_{i+1}\}, s_{i+1})$ is only added to $\antichain$ after $(\{s_i\}, s_i)$ has finished exploring its outgoing transitions.
This occurs in every state, because the state reached via such a transition was not visited before.
Hence, $\working$ contains exactly $i\cdot(k-1) + 1$ pairs at the end of the $i$-th iteration, resulting in a maximum $\working$ stack size of $\functionorder(n \cdot k)$ entries.
At the end of the $n$-th iteration $\antichain$ contains all reachable pairs of the product, \ie, $\antichain$ is equal to $(\{s_i\}, s_i)$ for all $1 \leq i \leq n$.
Emptying $\working$ after the $n$-th iteration involves $k$ antichain membership tests per entry.
Consequently, $\functionorder(n\cdot k^2)$ antichain membership tests are required to check $\lts_n^k \refinedbytrace \lts_n^k$.

The breadth-first variant of Algorithm~\ref{alg:original_tr} also adds the state reachable via a single step once for every action for the same reason as the depth-first variant.
However, now $(\{s_{i+1}\}, s_{i+1})$ is only added to the $\antichain$ after \emph{all} $k$ copies of $(\{s_i\}, s_i)$ are taken from the FIFO queue $\working$.
Therefore each entry in $\working$ adds $k$ elements before it is added to $\antichain$, resulting in a maximum queue size of $\functionorder(k^n)$ at state $(\{s_1\}, s_1)$. 
Emptying $\working$ results in $\functionorder(k^{n+1})$ antichain membership tests.\qed
\end{exa}

\noindent Finally, the example below illustrates the third problem of the algorithms, \viz, the violation of the antichain property.
We again illustrate the problem on the most basic of all three algorithms, \viz, Algorithm~\ref{alg:original_tr}.
Note that this violation does not influence the result of antichain membership tests, but it can have an effect on the size of the antichain which in turn leads to overhead.

\begin{exa}\label{example:violation}
  Consider the two left-most labelled transition systems depicted below,
  along with the (normal form) product (the LTS on the right).
  \begin{center}
    \scriptsize
    \begin{tikzpicture}[->, node distance=30pt]
    \tikzstyle{vertex}=[draw=none,minimum size=12pt,inner sep=0pt]
    
    \node[vertex] (T0)  {$t_0$};
    \node[vertex] (T0init) [above=10pt of T0] {};
    \node[vertex] (T1) [below left=30pt and 10pt of T0] {$t_1$};
    \node[vertex] (T2) [below right=30pt and 10pt of T0] {$t_2$};
    
    \draw (T0init) edge (T0);
    \draw (T0) edge[bend left] node[right]  {$a$} (T1) edge[bend right] node[left] {$b$} (T1);
    \draw (T0) edge node[right] {$b$} (T2);
    
    \node[vertex] (A0) [right=60pt of T0] {$s_0$};
    \node[vertex] (A0init) [above=10pt of A0] {};
    \node[vertex] (A1) [below=30pt of A0] {$s_1$};
    
    \draw (A0init) edge (A0);
    \draw (A0) edge[bend right=30] node[left]  {$a$} (A1);
    \draw (A0) edge[bend left=30] node[right] {$b$} (A1);
     
    \node[vertex] (P0) [right=2.5cm of A0] {$(\{t_0\},s_0)$};
    \node[vertex] (P0init) [above=10pt of P0] {};
    \node[vertex] (P1) [below left=1cm and 0cm of P0] {$\phantom{t_2}(\{t_1\},s_1)\phantom{,}$};
    \node[vertex] (P2) [below right=1cm and 0cm of P0] {$(\{t_1, t_2\},s_1)$};
    
    \draw (P0init) edge (P0);
    \draw (P0) edge node[above left]  {$a$} (P1);
    \draw (P0) edge node[above right] {$b$} (P2);  
    \end{tikzpicture} 
  \end{center}
  Algorithm~\ref{alg:original_tr} starts with $\working$ containing
  pair $(\{t_0\},s_0)$ and $\antichain = \emptyset$. 
  Inside the loop,
  the pair $(\{t_0\},s_0)$ is popped from $\working$ and added to
  $\antichain$. 
  The successors of the pair $(\{t_0\},s_0)$ are the pairs $(\{t_1\},s_1)$ and $(\{t_1, t_2\},s_1)$. 
  Since $\antichain$ contains neither of these, both successors are added to $\working$ in line~\ref{alg:original_tr:push_working}.
  Next, popping $(\{t_1\},s_1)$ from $\working$ and adding this pair to $\antichain$ results in $\antichain$ consisting of the set $\{(\{t_0\},s_0), ( \{t_1\},s_1)\}$. In the final iteration of the algorithm, the pair $(\{t_1, t_2\},s_1)$ is popped from $\working$ and  added to $\antichain$, resulting in the set $\{(\{t_0\},s_0), (\{t_1\},s_1), (\{t_1, t_2\},s_1)\}$. 
  Clearly, since $(\{t_1\},s_1) \statesubset ( \{t_1, t_2\},s_1)$, the set $\antichain$ no longer is a proper antichain.\qed  
\end{exa}

\section{Correct and Improved Antichain Algorithms}\label{section:improved-algorithm}

We first focus on solving the performance problems of Algorithms~\ref{alg:original_tr} and~\ref{alg:original_sfr}.
Subsequently, we discuss the additional modifications that are required for Algorithm~\ref{alg:original_fdr} to correctly decide failures-divergences refinement. 

The first performance problem that we identified, \viz, the computational overhead induced by checking for refusal inclusion in non-stable states (which does not occur when checking for a TR-witness), can be solved in a rather straightforward manner: we only perform the check to compare the refusals of the implementation and the normal form state of the specification in case the implementation state is stable.
Doing so avoids a potentially expensive search for stable states.

The second and third performance problems we identified can be solved by rearranging the computations that are conducted; these modifications are more involved.
The essential observation here is that in order for the information in $\antichain$ to be most effective, states of the product must be added to $\antichain$ as soon as these are discovered, even if these have not yet been fully explored.  
This is achieved by maintaining, as an invariant, that $\working \chainsetin \antichain$ holds true; the states in $\working$ then, intuitively, constitute the \emph{frontier} of the exploration.
We achieve this by initialising $\working$ and $\antichain$ to consist of exactly the initial state of the product, and by extending $\antichain$ with all (not already discovered) successors for the state $(\spec, \impl)$ that is popped from $\working$.
As a side effect, this also resolves the third issue, as now both $\working$ and $\antichain$ are only extended with states that have not yet been discovered, \ie, for which the membership test in $\antichain$ fails, and for which insertion of such states does not invalidate the
antichain property.

The modifications we discussed above yield improved algorithms for deciding trace refinement and stable failures refinement, see the pseudocode of Algorithms~\ref{alg:improved_tr} and~\ref{alg:improved_sfr}.
We postpone the discussion of their correctness until after discussing the modifications required to Algorithm~\ref{alg:original_fdr} and its proof of correctness.
The example we present below illustrates the impact of our changes.

\begin{algorithm}[ht]
  \footnotesize
  \caption{The improved trace refinement checking algorithm. 
  The algorithm returns \emph{true} iff $\lts_1 = (\states_1, \init_1, \transitions_1)$ is refined by $\lts_2  = (\states_2, \init_2, \transitions_2)$ in trace semantics.}
  \label{alg:improved_tr}
  \begin{algorithmic}[1]
    \Procedure{\refinesnew{trace}}{$\lts_1, \lts_2$}
    \State {let $\working$ be a stack containing a pair $(\{s \in \states_1 \mid \init_1 \weaktransition{\emptytrace}_1 s\}, \init_2)$ }
    \State {let $\antichain \gets \emptyset \chaininsert (\{s \in \states_1 \mid \init_1 \weaktransition{\emptytrace}_1 s\}, \init_2)$}
    \While {$\working \neq \emptyset$}
      \State {pop ($\spec, \impl$) from $\working$}
      \For {$\impl \transition{a}_2 \impl'$}
        \If {$a = \tau$}
          \State {$\spec' \gets \spec$}
        \Else
          \State {$\spec' \gets \{s' \in \states_1 \mid \exists s \in \spec : s \weaktransition{a}_1 s'\}$}		
        \EndIf
        \If {$\spec' = \emptyset$}
          \Return false
        \EndIf
        \If {$(\spec', \impl') \chainnotin \antichain$}
          \State {$\antichain \gets \antichain \chaininsert (\spec', \impl')$}
          \State {push $(\spec', \impl')$ into $\working$}
        \EndIf	
      \EndFor
    \EndWhile
    \Return true
    \EndProcedure
  \end{algorithmic}
\end{algorithm}

\begin{algorithm}[ht]
  \footnotesize
  \caption{The improved stable failures refinement checking algorithm. 
  The algorithm returns \emph{true} iff  $\lts_1 = (\states_1, \init_1, \transitions_1)$ is refined by $\lts_2  = (\states_2, \init_2, \transitions_2)$ in stable failures semantics.}
  \label{alg:improved_sfr}
  \begin{algorithmic}[1]
    \Procedure{\refinesnew{stable-failures}}{$\lts_1, \lts_2$}
    \State {let $\working$ be a stack containing a pair $(\{s \in \states_1 \mid \init_1 \weaktransition{\emptytrace}_1 s\}, \init_2)$ }
    \State {let $\antichain \gets \emptyset \chaininsert (\{s \in \states_1 \mid \init_1 \weaktransition{\emptytrace}_1 s\}, \init_2)$}
    \While {$\working \neq \emptyset$}
      \State {pop ($\spec, \impl$) from $\working$} 
          \If {$\stable(\impl) \land 
            \refusals(\impl) \not\subseteq \refusals(\spec)$}\label{alg:improved_sfr:refusals}
            \Return false
          \EndIf
          \For {$\impl \transition{a}_2 \impl'$}
            \If {$a = \tau$}
              \State {$\spec' \gets \spec$}
            \Else
              \State {$\spec' \gets \{s' \in \states_1 \mid \exists s \in \spec : s \weaktransition{a}_1 s'\}$}	
            \EndIf
            \If {$\spec' = \emptyset$}
            \Return false
            \EndIf
            \If {$(\spec', \impl') \chainnotin \antichain$}
              \State {$\antichain \gets \antichain \chaininsert (\spec', \impl')$}
              \State {push $(\spec', \impl')$ into $\working$}
            \EndIf	
          \EndFor
    \EndWhile
    \Return true
    \EndProcedure
  \end{algorithmic}
\end{algorithm}

\begin{exa}
  \label{label:analysis}Consider Example~\ref{example:badcase} again, but now using Algorithm~\ref{alg:improved_tr} to check for trace refinement.
  The depth-first variant of this algorithm only adds a successor state to the $\working$ stack once, because for every other outgoing transition it will already be part of $\antichain$ when it is discovered.
  This results in a maximum $\working$ stack size of at most $\functionorder(1)$ entries.
  For each state and each successor $\antichain$ membership is tested once, resulting in $\functionorder(n\cdot k)$ $\antichain$ membership tests.
  This is an improvement compared to the depth-first variant of Algorithm~\ref{alg:improved_tr} of a factor $n\cdot k$ in the maximum $\working$ stack size and a factor $k$ in the number of $\antichain$ membership tests.
  The bounds for the breadth-first variant are identical to the bounds for the depth-first variant, \ie, maximum $\functionorder(1)$ $\working$ queue size and $\functionorder(n\cdot k)$ number of $\antichain$ membership tests.
  Compared to the breadth-first variant of Algorithm~\ref{alg:original_tr}, this is an improvement of a factor $k^n$ in the $\working$ queue size and a factor $k^n/n$ in the number of $\antichain$ membership tests.\qed
\end{exa}

\begin{algorithm}[ht]
  \footnotesize
  \caption{The corrected failures-divergences refinement checking algorithm. 
  The algorithm returns \emph{true} iff $\lts_1 = (\states_1, \init_1, \transitions_1)$ is refined by $\lts_2  = (\states_2, \init_2, \transitions_2)$ in failures-divergences semantics.}
  \label{alg:improved_fdr}
  \begin{algorithmic}[1]
    \Procedure{\refinesnew{failures-divergences}}{$\lts_1, \lts_2$}
    \State {let $\working$ be a stack containing a pair $(\{s \in \states_1 \mid \init_1 \weaktransition{\emptytrace}_1 s\}, \init_2)$ }
    \State {let $\antichain \gets \emptyset \chaininsert (\{s \in \states_1 \mid \init_1 \weaktransition{\emptytrace}_1 s\}, \init_2)$}
    \State {$\{~\Done \gets \emptyset ~\}$}\label{alg:improved_fdr:init_done}
    \While {$\working \neq \emptyset$}\label{alg:improved_fdr:while-working}
      \State {pop ($\spec, \impl$) from $\working$}\label{alg:improved_fdr:pop_working}
      \If {not $\diverges{\spec}$}\label{alg:improved_fdr:spec}
        \If {$\diverges{\impl}$}\label{alg:improved_fdr:divergences}
          \Return false
        \Else
          \If {$\stable(\impl) \land \refusals(\impl) \not\subseteq \refusals(\spec)$}\label{alg:improved_fdr:refusals}
            \Return false
          \EndIf
          \For {$\impl \transition{a}_2 \impl'$}\label{alg:improved_fdr:transition}
            \If {$a = \tau$}\label{alg:improved_fdr:case_tau}
              \State {$\spec' \gets \spec$}
            \Else\label{alg:improved_fdr:case_event}
              \State {$\spec' \gets \{s' \in \states_1 \mid \exists s \in \spec : s \weaktransition{a}_1 s'\}$}	
            \EndIf
            \If {$\spec' = \emptyset$}\label{alg:improved_fdr:unreachable}
              \Return false
            \EndIf
            \If {$(\spec', \impl') \chainnotin \antichain$}\label{alg:improved_fdr:inclusion}
              \State {$\antichain \gets \antichain \chaininsert (\spec', \impl')$}\label{alg:improved_fdr:add-antichain}
              \State {push $(\spec', \impl')$ into $\working$}\label{alg:improved_fdr:push_working}
            \EndIf	
          \EndFor
        \EndIf
      \EndIf
      \State {$\{~\Done \gets \{(\spec, \impl)\} \cup \Done~ \}$}\label{alg:improved_fdr:update_done}
    \EndWhile\label{alg:improved_fdr:while-working-end}
    \Return true
    \EndProcedure
  \end{algorithmic}
\end{algorithm}

\noindent We next focus on the soundness problem of Algorithm~\ref{alg:original_fdr}.
The source of the incorrectness of this algorithm can be traced back to the fact that it (partially) explores the state space of $\normalizesfr(\lts_1) \product \lts_2$, rather than $\normalize(\lts_1) \product \lts_2$.
As illustrated by Example~\ref{example:normalizefdr}, this causes the algorithm to consider states in the product that should not be considered, thus potentially arriving at a wrong verdict.
The fix to this problem is simple yet subtle, requiring a  swap of the divergence tests on lines~\ref{alg:original_fdr:impl_diverges} and~\ref{alg:original_fdr:spec_diverges}, and making the further exploration of the state $(\spec,\impl)$ conditional on the specification not diverging.

As all three algorithms presented in this section fundamentally differ (some even in the relations that they compute) from the original ones, we cannot reuse arguments for the proof of correctness presented in~\cite{WangS0LDWL12}, which are based on invariants that do not hold in our case, and which rely on definitions, some of which are incomparable to ours.
The correctness of our improved algorithms is claimed by the following theorem, which we repeat at the end of this section with an explicit proof.

\begin{restatable}{thm}{refinementtheorem}\label{theorem:algorithm}
  Let $\lts_1 = (\states_1,\init_1,\transitions_1)$ and $\lts_2 = (\states_2,\init_2,\transitions_2)$ be two LTSs. 

  \begin{itemize}
    \item \refinesnew{trace}($\lts_1$, $\lts_2$) returns true if and only if $\lts_1 \refinedbytrace \lts_2$.

    \item \refinesnew{stable-failures}($\lts_1$, $\lts_2$) returns true if and only if $\lts_1 \refinedbysfr \lts_2$.
    
    \item \refinesnew{failures-divergences}($\lts_1$, $\lts_2$) returns true if and only if $\lts_1 \refinedbyfdr \lts_2$.
  \end{itemize}
\end{restatable}

\noindent For the remainder of this section we fix two LTSs $\lts_1 = (\states_1,\init_1,\transitions_1)$ and $\lts_2 = (\states_2,\init_2,\transitions_2)$.
We focus on the proof of correctness of Algorithm~\ref{alg:improved_fdr}; the correctness proofs for Algorithm~\ref{alg:improved_tr} for deciding trace refinement and Algorithm~\ref{alg:improved_sfr} for deciding stable failures refinement proceed along the same lines.

First we show termination of Algorithm~\ref{alg:improved_fdr}.
To reason about the states that have been processed, we have introduced a ghost variable $\Done$ which is initialised as the empty set (see line~\ref{alg:improved_fdr:init_done}) and each pair $(\spec, \impl)$ that is popped from $\working$ at line~\ref{alg:improved_fdr:pop_working} is added to $\Done$ (line~\ref{alg:improved_fdr:update_done}).
For termination of the algorithm, we argue that every state in the product gets visited, and is added to $\Done$, at most once.
A crucial observation in our reasoning is the following property of an antichain: adding elements to an antichain does not affect the membership test of elements already included.  
This is formalised by the lemma below.

\begin{lem}\label{lemma:consistent}
  Let  $(Z, \leq)$ be a partially ordered set and $\Antichain \subseteq Z$ an antichain.
  For all elements $x, y \in Z$ if $x \chainin \Antichain$ and $y \chainnotin \Antichain$ then $x \chainin (\Antichain \chaininsert y)$ holds. 
\end{lem}

\begin{proof}  
  Assume arbitrary elements $x, y \in Z$ such that $x \chainin \Antichain$ and $y \chainnotin \Antichain$.
  Recall that the definition of $\Antichain \chaininsert y$ results in an antichain $\{ z \mid z = y \vee (z \in \mathcal{A} \wedge y \not\leq z)\}$, because $y \chainnotin \Antichain$ by assumption.
  Consider the following two cases:
  \begin{itemize}
    \item Case $y \statesubset x$.
    Then $x \chainin (\Antichain \chaininsert y)$ follows from the fact that $y \in (\Antichain \chaininsert y)$. 
    
    \item Case $y \statensubset x$.
    There is an element $z \in \Antichain$ such that $z \statesubset x$ by assumption that $x \chainin \Antichain$.
    Because $y \statensubset x$ and $z \statesubset x$ we also know that $y \statensubset z$.
    Consequently, $z \in (\Antichain \chaininsert y)$ and thus also $x \chainin (\Antichain \chaininsert y)$.\qedhere
  \end{itemize}
\end{proof}

\noindent Next, we prove that $\Done$ and $\working$ are disjoint, which implies that pairs present in $\Done$ (which is a set that is easily seen to only grow) are not added to $\working$ again.
Showing this property to be true requires two additional observations, \viz, (1) pairs in $\working$ and $\Done$ are contained in $\antichain$, and (2), $\working$ contains only \emph{unique} pairs, thus representing a proper set (and, by abuse of notation, we will treat it as such).
For the purpose of identifying elements in $\working$ we define, for a given index $i$, the notation $\working^i$ to represent the $i$th pair on the stack.
Now we can describe that all elements in $\working$ are unique by showing that $\forall i \neq j : \working^i \neq \working^j$ holds true.
The lemma below formalises these insights.

\begin{lem}\label{lemma:disjoint}
  The following invariant holds in the while loop (lines~\ref{alg:improved_fdr:while-working}-\ref{alg:improved_fdr:while-working-end}) of Algorithm~\ref{alg:improved_fdr}:
  \begin{equation}\tag{\text{I}}\label{invariant:working}
    (\Done \cup \working) \chainsetin \antichain \land (\forall i \neq j : \working^i \neq \working^j) \land (\Done \cap \working) = \emptyset
  \end{equation}
\end{lem}

\begin{proof}
  Initially, the initial pair is added to both $\working$ and $\antichain$, and $\Done$ is empty, so the invariant holds trivially upon entry of the while loop.
  
  Maintenance.
  At line~\ref{alg:improved_fdr:pop_working} we know that $(\spec, \impl) \chainin \antichain$ from $\working \chainsetin \antichain$.
  Therefore, it holds that $(\Done \cup \{(\spec,\impl)\}) \chainsetin \antichain$ and $(\working \setminus \{(\spec, \impl)\}) \chainsetin \antichain$.
  Furthermore, from $\forall i \neq j : \working^i \neq \working^j$ it follows that $(\spec, \impl) \notin (\working \setminus \{(\spec, \impl)\})$.
  Upon executing line~\ref{alg:improved_fdr:pop_working} we may therefore conclude that
  $((\Done \cup \{(\spec,\impl)\}) \cap (\working \setminus \{(\spec, \impl)\})) = \emptyset$.
        
  Next, notice that as a result of condition $(\spec', \impl') \chainnotin \antichain$ on line~\ref{alg:improved_fdr:inclusion}, we have $(\spec', \impl') \notin (\Done \cup \working)$.
  Let $\working' = \{(\spec', \impl')\} \cup (\working \setminus \{(\spec, \impl)\})$ and $\Done' = \{(\spec,\impl)\} \cup \Done$.
  From the fact that $(\spec', \impl') \notin \working$ it follows that $\forall i \neq j : \working'^i \neq \working'^j$ holds true.
  At line~\ref{alg:improved_fdr:add-antichain} the $(\spec', \impl')$ pair is added to $\antichain$ and Lemma~\ref{lemma:consistent} ensures that $(\Done' \cup \working') \chainsetin (\antichain \chaininsert (\spec', \impl'))$ holds.
  Finally, from $((\Done \cup \{(\spec,\impl)\}) \cap (\working \setminus \{(\spec, \impl)\})) = \emptyset$ we can also conclude that $(\Done' \cap \working') = \emptyset$. 
\end{proof}

\noindent Finally, we need to show that the elements in $\Done$ and $\working$ are bounded by the state space of the product $\normalize(\lts_1) \product \lts_2$.
 
\begin{lem}\label{lemma:reachable}
  The invariant $(\Done \cup \working) \subseteq \reachable(\normalize(\lts_1) \product \lts_2)$ holds in the while loop (lines~\ref{alg:improved_fdr:while-working}-\ref{alg:improved_fdr:while-working-end}) of Algorithm~\ref{alg:improved_fdr}.
\end{lem}
\begin{proof}
  Initially, the pair $(\{s \in \states_1 \mid \init_1 \weaktransition{\emptytrace}_1 s\}, \init_2)$ is reachable by the empty trace because this pair is the initial state of $\normalize(\lts_1) \product \lts_2$ by definition. 
  Therefore, $\working$, which only consists of this pair, contains pairs that are reachable as well.
  Moreover, $\Done$ is empty, so the invariant holds upon entry of the while loop.
  
  Maintenance. 
  Let $(\spec,\impl)$ be a pair that is popped from $\working$ and assume that $\diverges{\tostates{\spec}_{\lts_1}}$ does not hold.
  Note that, by our invariant, there is a trace $\trace \in \actionstau^*$, such that $\init \tracetransition{\trace} (\spec, \impl)$.
  At line~\ref{alg:improved_fdr:transition} the outgoing transition $(\impl, a, \impl')$ is an element of $\transitions_2$.
  Line~\ref{alg:improved_fdr:case_tau} corresponds exactly to the first case of the product definition (Def.~\ref{def:product}).
  Similarly, line~\ref{alg:improved_fdr:case_event} corresponds exactly to the second case of the product definition where $(\spec, a, \spec')$ is a transition in $\normalize(\lts_1)$ because $\diverges{\tostates{\spec}_{\lts_1}}$ does not hold.
  As such, there is a transition $(\spec, \impl) \transition{a} (\spec', \impl')$ in the product LTS.
  By definition of a trace and the definition of reachable this means that $(\working \cup \{(\spec', \impl')\}) \subseteq \reachable(\normalize(\lts_1) \product \lts_2)$.
  From the observation that $(\spec, \impl)$ was reachable we can conclude that $(\Done \cup \{(\spec, \impl)\})$ is a subset of the reachable states as well.
\end{proof}

\begin{thm}\label{theorem:termination}
  Algorithm~\ref{alg:improved_fdr} terminates for finite state, finitely branching LTSs.
\end{thm}

\begin{proof}
  The inner for-loop is bounded as the number of outgoing transitions $\transitions_2$ is finite. 
  The total number of state pairs in $\normalize(\lts_1) \product \lts_2$ is finite since $\states_1$ and $\states_2$ are finite.
  From Lemma~\ref{lemma:reachable} it follows that $\Done$ is a subset of the reachable state pairs.
  Furthermore, as $(\Done \cap \working) = \emptyset$ by Lemma~\ref{lemma:disjoint} we conclude that $\Done$ strictly increases with every iteration.
  So, only a finite number of iterations of the while loop are possible.
\end{proof}

\noindent Note that these observations give an upper bound on the number of states that can be explored.
Especially the absence of duplicates in $\working$ and the maximisation of $\antichain$ following from $(\Done \cup \working) \chainsetin \antichain$ do not hold for Algorithm~\ref{alg:original_tr},~\ref{alg:original_sfr} and~\ref{alg:original_fdr}, as already observed in Example~\ref{example:badcase}.

The remainder of this section is dedicated to proving the partial correctness of Algorithm~\ref{alg:improved_fdr}, \viz, that when it terminates, the algorithm correctly decides failures-divergences refinement.
We first revisit the claim we made in Section~\ref{section:original-algorithm}; before we restate and prove this claim, we prove a simplified version thereof in the next lemma.

\begin{lem}\label{lemma:subset_transitions}
  For all states $(U, s), (V, s)$ of $\normalize(\lts_1) \product \lts_2$ satisfying $(U, s) \statesubset (V, s)$ and actions $a \in \actionstau$ such that $(V, s) \transition{a} (V', t)$ there is a state $(U', t)$ such that $(U, s) \transition{a} (U', t)$ and $(U', t) \statesubset (V', t)$. 
\end{lem}

\begin{proof}
  Let $\normalize(\lts_1) \product \lts_2 = (\states, \init, \transitions)$ and let $\normalize(\lts_1) = (\states_1', \init_1', \transitions_1')$.
  Take any two state pairs such that $(U, s) \statesubset (V, s)$.
  Pick an arbitrary pair $(V', t) \in \states$ and action $a \in \actionstau$ such that $(V, s) \transition{a} (V', t)$.
  Now there are two cases to distinguish:
  
  \begin{itemize}
  \item Case $a = \tau$.
    Then a transition $s \transition{\tau}_2 t$ exists and $V = V'$.
    Therefore, there is also a transition $(U, s) \transition{\tau} (U', t)$ and $U = U'$.
    By the assumption that $(U, s) \statesubset (V, s)$ we know that $(U', t) \statesubset (V', t)$.
  
  \item Case $a \neq \tau$.
    Then there are transitions $V \transition{a}_1' V'$ and $s \transition{a}_2 t$.
    The normalisation has, by definition, transition $V \transition{a}_1' V'$ if and only if $V' = \{v' \in \states_1 \mid \exists v \in \tostates{V}_{\lts_1} : v \weaktransition{a}_1 v'\}$ and not $\diverges{\tostates{V}_{\lts_1}}$.
    Let $U'$ be equal to $\{u' \in \states_1 \mid \exists u \in \tostates{U}_{\lts_1} : u \weaktransition{a}_1 u'\}$.
    From $\tostates{U}_{\lts_1} \subseteq \tostates{V}_{\lts_1}$ it follows that $\tostates{U'}_{\lts_1} \subseteq \tostates{V'}_{\lts_1}$.
    Furthermore, as $\diverges{\tostates{V}_{\lts_1}}$ does not hold it follows that not $\diverges{\tostates{U}_{\lts_1}}$.    
    Therefore, $U \transition{a}_1' U'$ exists and $(U, s) \transition{a} (U', t)$ is a transition in the product with $(U', t) \statesubset (V', t)$.\qedhere
  \end{itemize}
\end{proof}

\noindent We are now in a position to formally prove the claim that we made in Section~\ref{section:original-algorithm}.
For convenience, we repeat the claim as a proposition below.

\begin{prop}\label{lemma:subset_trace}
  For all states $(U, s), (V, s)$ of $\normalize(\lts_1) \product \lts_2$ satisfying $(U, s) \statesubset (V, s)$ and for every sequence $\trace \in \actionstau^*$ such that $(V, s) \tracetransition{\trace} (V', t)$ there is a state $(U',t)$ such that  $(U, s) \tracetransition{\trace} (U', t)$ and $(U', t) \statesubset (V', t)$.
\end{prop}

\begin{proof}
  Let $\normalize(\lts_1) \product \lts_2 = (\states, \init, \transitions)$.
  The proof is by induction on the length of sequences in $\actionstau^*$.
    
  Base case.
  Take two pairs $(U, s), (V, s) \in \states$ satisfying $(U, s) \statesubset (V, s)$.
  The empty trace can only reach $(U, s) \tracetransition{\emptytrace} (U, s)$; similarly, we have $(V,s) \tracetransition{\emptytrace} (V,s)$.
  Therefore,  $(U, s) \statesubset (V, s)$ follows by assumption.
    
  Inductive step.
  Suppose that the statement holds for all sequences in $\actionstau^*$ of length $i$ and take a sequence $\trace \in \actionstau^*$ of length $i$.
  Take arbitrary states $(V, s), (V'', r) \in \states$ and action $a \in \actionstau$ such that $(V, s) \tracetransition{\trace\,a} (V'', r)$.
  Then there is a state $(V', t) \in \states$ such that  $(V, s) \tracetransition{\trace} (V', t)$ and $(V', t) \transition{a} (V'', r)$.  
  From the induction hypothesis it follows that for all $(U, s) \statesubset (V, s)$ there is a state $(U', t) \statesubset (V', t)$ such that $(U, s) \tracetransition{\trace} (U', t)$.
  By Lemma~\ref{lemma:subset_transitions} and the existence of $(V', t) \transition{a} (V'', r)$ there is a state $(U'', r) \statesubset (V'', r)$ such that $(U', t) \transition{a} (U'', r)$.
  We thus conclude that $(U, s) \tracetransition{\trace\,a} (U'', r)$.\qedhere
\end{proof}

\noindent The correctness arguments of Algorithm~\ref{alg:improved_fdr} furthermore require a lemma showing the anti-monotonicity of FD-witnesses.
Such a result is needed because the antichain algorithms may explore only part of the reachable state space of a product.
The anti-monotonicity property helps to show, however, that the part that \emph{is} explored contains all relevant information.

\begin{lem}\label{lemma:fdrwitnes_subset}
  For all states $(U, s), (V, s)$ of $\normalize(\lts_1) \product \lts_2$ satisfying $(U, s) \statesubset (V, s)$ it holds that if $(V, s)$ is an FD-witness then $(U, s)$ is an FD-witness.
\end{lem}
\begin{proof}
  Take arbitrary states $(U, s), (V, s)$ of $\normalize(\lts_1) \product \lts_2$ satisfying $(U, s) \statesubset (V, s)$ and let $(V, s)$ be an FD-witness.
  It follows that $\diverges{\tostates{V}_{\lts_1}}$ does not hold and one of the following holds: $V = \emptyset$ or $\stable(s) \land \refusals(s) \nsubseteq \refusals(\tostates{V}_{\lts_1})$ or $\diverges{s}$.
  By monotonicity, $\tostates{U}_{\lts_1} \subseteq \tostates{V}_{\lts_1}$ implies $\refusals(\tostates{U}_{\lts_1}) \subseteq \refusals(\tostates{V}_{\lts_1})$, and not $\diverges{\tostates{V}_{\lts_1}}$ implies not $\diverges{\tostates{U}_{\lts_1}}$.
  Now, $(U, s)$ is an FD-witness, because $\diverges{\tostates{U}_{\lts_1}}$ does not hold and if $V = \emptyset$ then $U = \emptyset$, or if $\stable(s) \land \refusals(s) \nsubseteq \refusals(\tostates{V}_{\lts_1})$ then $\stable(s) \land \refusals(s) \nsubseteq \refusals(\tostates{U}_{\lts_1})$, or $\diverges{s}$.
\end{proof}

\begin{cor}\label{lemma:fdrwitness_trace}
  For all states $(U, s), (V, s)$ of $\normalize(\lts_1) \product \lts_2$ where $(U, s) \statesubset (V, s)$ and for every sequence $\trace \in \actionstau^*$ it holds that if $(V, s)$ can reach an FD-witness with $\trace$ then $(U, s)$ can reach an FD-witness with $\trace$ as well.
\end{cor}
\begin{proof}
  Let $\normalize(\lts_1) \product \lts_2 = (\states, \init, \transitions)$.
  Take arbitrary states $(U, s), (V, s) \in \states$ satisfying $(U, s) \statesubset (V, s)$.
  Let $(V', t)$ be an FD-witness and $\trace \in \actionstau^*$ a trace such that $(V, s) \tracetransition{\trace} (V', t)$.
  By Lemma~\ref{lemma:subset_trace} there is a pair $(U', t) \statesubset (V', t)$ such that $(U, s) \tracetransition{\trace} (U', t)$.
  From Lemma~\ref{lemma:fdrwitnes_subset} it follows that state $(U', t)$ is an FD-witness.
\end{proof}

\noindent For a set of states $S'$ of $\normalize(\lts_1) \product \lts_2$, let $\FDR(S')$ be the predicate that is true if and only if $S'$ contains an FD-witness. 
For a state $s$ in the product, we define the \emph{distance} to a set of states $S'$ of the product as the \emph{shortest} distance from state $s$ to a state in $S'$.
If $S'$ is unreachable, the distance is set to infinity.
Formally, $\Dist_{S'}(s) = \min\{\length{\trace} \mid \trace \in \traces(\lts) \land t \in {S'} \land s \tracetransition{\trace} t\}$, where $\min\{\emptyset\}$ is defined as $\infty$.
For a set of states $S''$, let $\Dist_{S'}(S'')$ denote the shortest distance among all states in $S''$, formally $\Dist_{S'}(S'') = \min\{\Dist_{S'}(s) \mid s \in S'' \}$.
We denote the set of all reachable FD-witnesses in the product $\normalize(\lts_1) \product \lts_2$ by $\witnesses$.

\begin{lem}\label{lemma:fdrwitness_distance}
  For all states $(U, s), (V, s)$ of $\normalize(\lts_1) \product \lts_2$ satisfying $(U, s) \statesubset (V, s)$ it holds that $\Dist_\witnesses((U, s)) \leq \Dist_\witnesses((V, s))$.
\end{lem}
\begin{proof}
  Let $\normalize(\lts_1) \product \lts_2 = (\states, \init, \transitions)$.
  Take arbitrary states $(U, s), (V, s) \in \states$ satisfying $(U, s) \statesubset (V, s)$.
  From Corollary~\ref{lemma:fdrwitness_trace} it follows that if $(V, s)$ can reach an FD-witness by the shortest trace $\trace$ then $(U, s)$ can also reach an FD-witness with trace $\trace$, which by definition means that $\Dist_\witnesses((U, s)) \leq \Dist_\witnesses((V, s))$.
\end{proof}

\noindent The last lemma implies that whenever a pair is removed from the $\antichain$ due to an insertion of a smaller pair, the inserted (smaller) state pair has a shorter or equal distance to its closest FD-witness.
This property can be used to show that the algorithm always closes in on an FD-witness during exploration and that pruning parts of the state space does not remove essential FD-witnesses from the reachable states.
The latter property is captured by the following lemmas.

\begin{lem}\label{lemma:diverging_spec}
  For all states $(U, s)$ of $\normalize(\lts_1) \product \lts_2$ it holds that if $\diverges{\tostates{U}_{\lts_1}}$ then $\Dist_\witnesses((U, s))$ is $\infty$.
\end{lem}

\begin{proof}
  Let $\normalize(\lts_1) \product \lts_2 = (\states, \init, \transitions)$ and let $\normalize(\lts_1) = (\states_1', \init_1', \transitions_1')$.
  Take an arbitrary state $(U, s) \in \states$ such that $\diverges{\tostates{U}_{\lts_1}}$.
  For any action $a \in \actionstau$ and state $V \in \states_1'$ there is no 
  transition $U \transition{a}_1' V$ by definition of $\normalize$.
  Consequently, from $(U,s)$, only $\tau$-transitions due to $\lts_2$ can be taken.
  As a result, by definition of the product and Lemma~\ref{lemma:product_traces}, for any state $(V, t) \in \states$ such that $(U, s) \weaktransition{\emptytrace} (V, t)$ it holds that $U = V$.
  Thus, any reachable state $(V, t)$ also satisfies $\diverges{\tostates{V}_{\lts_1}}$ and therefore cannot be an FD-witness.
  Hence, $\Dist_\witnesses((U, s))$ is $\infty$.
\end{proof}

\begin{lem}\label{lemma:invariant_distance}
  If $\FDR(\reachable(\normalize(\lts_1) \product \lts_2))$ is true then invariant~\ref{invariant:distance} holds for every iteration of the while loop (lines~\ref{alg:improved_fdr:while-working}-\ref{alg:improved_fdr:while-working-end}) of Algorithm~\ref{alg:improved_fdr}:  
  \begin{equation}\tag{\text{II}}\label{invariant:distance}
    \Dist_\witnesses(\Done) > \Dist_\witnesses(\working) \land \Dist_\witnesses(\working) = \Dist_\witnesses(\antichain)
  \end{equation}
\end{lem}

\begin{proof}
  Assume that $\FDR(\reachable(\normalize(\lts_1) \product \lts_2))$ holds, so there is a reachable FD-witness.
  
  Initialisation. 
  The set $\Done$ is empty, so $\Dist_\witnesses(\Done) = \Dist_\witnesses(\emptyset) = \infty$. 
  For $\working$, which at this point only contains the initial state, the witness is reachable and therefore $\Dist_\witnesses(\working) < \infty$.
  The initial state is also added to $\antichain$.
  Thus $\Dist_\witnesses(\Done) > \Dist_\witnesses(\working) \land \Dist_\witnesses(\working) = \Dist_\witnesses(\antichain)$.
  
  Maintenance. 
  Assume that $\working$ is not empty and that $\Dist_\witnesses(\Done) > \Dist_\witnesses(\working)$ and $\Dist_\witnesses(\working) = \Dist_\witnesses(\antichain)$ hold.
  At line~\ref{alg:improved_fdr:pop_working} a pair $(\spec, \impl)$ is taken from $\working$, so $\working$, which by invariant~I represents a set, becomes equal to $\working \setminus \{(\spec, \impl)\}$.
  Let $\Done' = \Done \cup \{(\spec, \impl)\}$ and let $N = \Dist_\witnesses((\spec, \impl))$.
  There are three cases to distinguish.
  
  \begin{itemize}
  \item Case $N > \Dist_\witnesses(\working \cup \{(\spec, \impl)\})$. 
    Removing $(\spec, \impl)$ from $\working$ did not change its distance, so $\Dist_\witnesses(\working) = \Dist_\witnesses(\working \cup \{(\spec, \impl)\})$.
    Because $N > \Dist_\witnesses(\working)$, adding this pair to $\Done$ results in $\Dist_\witnesses(\working) < \Dist_\witnesses(\Done') \leq \Dist_\witnesses(\Done)$.
    Consider the outgoing transitions $(\impl,a,\impl') \in\,\transitions_2$ at line~\ref{alg:improved_fdr:transition}.
    The resulting pairs $(\spec', \impl')$ must have a distance of at least $\Dist_\witnesses(\working)$, because $N - 1 \geq \Dist_\witnesses(\working)$.
    Let $\working' = \working \cup \{(\spec', \impl')\}$.
    Then $\Dist_\witnesses(\working) = \Dist_\witnesses(\working')$.    
    Let $\antichain'$ be $\antichain$ if $(\spec', \impl')$ was not inserted and let it be $\antichain \chaininsert (\spec', \impl')$ otherwise.
    By the invariant it follows that $N - 1 \geq \Dist_\witnesses(\antichain)$ and so by Lemma~\ref{lemma:fdrwitness_distance}
    if $(\spec', \impl')$ is inserted into $\antichain$ its distance will also not change.
    Therefore, $\Dist_\witnesses(\Done') > \Dist_\witnesses(\working') \land \Dist_\witnesses(\working') = \Dist_\witnesses(\antichain')$.
     
  \item Case $0 < N \leq \Dist_\witnesses(\working \cup \{(\spec, \impl)\})$.
    Observe that $N$ must be equal to $\Dist_\witnesses(\working \cup \{(\spec, 
    \impl)\})$.
    From Lemma~\ref{lemma:diverging_spec} it follows that $\diverges{\tostates{\spec}_{\lts_1}}$ does not hold and so the successors of $(\spec, \impl)$ are explored.
    Invariant~\ref{invariant:distance} holds upon termination of the inner for-loop at lines~13 to~22.
    This follows from an invariant for the inner for-loop, which we state next.

    Let $\antichain'$ be equal to the value of variable $\antichain$ after line~13 at each iteration and $\working'$ be equal to the value of variable $\working$.
    Furthermore, let $T$ be the set of successors of $(\spec, \impl)$, due to the transitions emanating from $\impl$, and let $\Done'$ be the successors that have been processed, \ie, $\Done'$ is initially empty and $(\spec', \impl')$ is inserted into it after line~17.
		  It can be shown, using Lemma~\ref{lemma:fdrwitness_distance}, that the following is an invariant for the inner for-loop: 
    \begin{align*}
      (\Dist_\witnesses(T \setminus \Done') < N \lor \Dist_\witnesses(\working') < N) \\ \land \Dist_\witnesses(\working') = \Dist_\witnesses(\antichain')
    \end{align*}
    
    Upon termination we conclude that $(T \setminus \Done') = \emptyset$.
    It then follows that $\Dist_\witnesses(T \setminus \Done') = \infty$.
    As a consequence, we find that $\Dist_\witnesses(\working') < N$ and therefore $\Dist(\working') < \Dist_\witnesses(\Done \cup \{(\spec, \impl)\})$.    
    
  \item Case $N = 0$. 
    The state $(\spec, \impl)$ is checked for the FD-witness conditions and the algorithm terminates.\qedhere
  \end{itemize}
\end{proof}

\noindent We conclude with the following result, which underlies the correctness of Algorithm~\ref{alg:improved_fdr}.

\begin{thm}\label{theorem:fdrwitness}
  Algorithm~\ref{alg:improved_fdr} returns false if and only if an FD-witness is reachable in the product $\normalize(\lts_1) \product \lts_2$.
\end{thm}

\begin{proof}\leavevmode
  \begin{itemize}[align=left,beginpenalty=999]
    
  \item[$(\Longrightarrow$)]
    Assume that Algorithm~\ref{alg:improved_fdr} returns false. 
    This occurs when the current pair $(\spec, \impl)$ satisfies the conditions of an FD-witness, as shown in lines~\ref{alg:improved_fdr:spec}, ~\ref{alg:improved_fdr:divergences}, \ref{alg:improved_fdr:refusals} and \ref{alg:improved_fdr:unreachable} of Algorithm~\ref{alg:improved_fdr}. All pairs taken from $\working$ are reachable according to Lemma~\ref{lemma:reachable}, so this FD-witness is also reachable.
  
  \item[($\Longleftarrow$)]
    Assume that an FD-witness is reachable in the product of $\normalize(\lts_1) \product \lts_2$, \ie, $\witnesses \neq \emptyset$.
    Then invariant~II of Lemma~\ref{lemma:invariant_distance} holds:    
    \begin{equation*}
      \Dist_\witnesses(\Done) > \Dist_\witnesses(\working) \land \Dist_\witnesses(\working) = \Dist_\witnesses(\antichain)
    \end{equation*}
    
    \noindent Towards a contradiction, assume that Algorithm~\ref{alg:improved_fdr} returns true.
    The algorithm returns true if and only if $\working$ is empty, which means that $\Dist_\witnesses(\working) = \Dist_\witnesses(\emptyset) = \infty$.
    The initial state $\init$ of $\normalize(\lts_1) \product \lts_2$ is equal to $(\{s \in \states_1 \mid \init_1 \weaktransition{\emptytrace}_1 s\}, \init_2)$ and can reach an FD-witness by assumption.
    Therefore, $\Dist_\witnesses(\init) < \infty$.
    Initially $\init$ was inserted into $\antichain$ so by Lemma~\ref{lemma:consistent} follows that $\init \chainin \antichain$ and from Lemma~\ref{lemma:fdrwitness_distance} it follows that $\Dist_\witnesses(\antichain) < \infty$. Contradiction, so we conclude that if Algorithm~\ref{alg:improved_fdr} terminates then it returns false.
    Since termination is shown in Theorem~\ref{theorem:termination} we establish that the algorithm returns false.\qedhere
  \end{itemize}
\end{proof}

\noindent
We here note that analogues of Theorem~\ref{theorem:fdrwitness} for
Algorithms~\ref{alg:improved_tr} and~\ref{alg:improved_sfr} can be proved along the same lines.
In particular, invariants~I and~II, fundamental in proving termination, and proved in Lemmas~\ref{lemma:disjoint} and~\ref{lemma:reachable}, can be shown to hold for both algorithms using the same arguments (where, of course, the counterpart of invariant~II relies on a distance to the set of TR-witnesses or SF-witnesses).
Proposition~\ref{lemma:subset_trace} also holds for the product $\normalizesfr(\lts_1) \product \lts_2$, and Lemma~\ref{lemma:fdrwitnes_subset} and Corollary~\ref{lemma:fdrwitness_trace} hold for TR-witnesses and SF-witnesses in the product $\normalizesfr(\lts_1) \product \lts_2$.
Without going into these details, we here claim the correctness for Algorithms~\ref{alg:improved_tr} and~\ref{alg:improved_sfr}. 
\begin{thm}\label{theorem:sfrwitness}
  Algorithm~\ref{alg:improved_tr} returns false if and only if a TR-witness is reachable in the product $\normalizesfr(\lts_1) \product \lts_2$. Algorithm~\ref{alg:improved_sfr} returns false if and only if an SF-witness is reachable in the product $\normalizesfr(\lts_1) \product \lts_2$.
\end{thm}

\noindent
We finish with restating the formal claim of correctness of all three improved algorithms.

\refinementtheorem*

\begin{proof}
  From Theorem~\ref{theorem:fdrwitness} we can conclude that Algorithm~\ref{alg:improved_fdr} returns false if and only if an FD-witness is reachable.
  By Theorem~\ref{theorem:failures} an FD-witness is only reachable if and only if $\lts_1$ does not refine $\lts_2$ in failures-divergences semantics.
  Virtually the same arguments apply for trace and stable failures refinement.
\end{proof}

\section{Experimental Validation}\label{section:experiments}

We have conducted several experiments to compare the run time of the various algorithms to show that solving the identified issues actually improves the run time performance in practice.
For this purpose we have implemented a depth-first and breadth-first variant for each of the original algorithms (Algorithms~\ref{alg:original_tr},~\ref{alg:original_sfr} and~\ref{alg:original_fdr}) and improved algorithms (Algorithms~\ref{alg:improved_tr},~\ref{alg:improved_sfr} and~\ref{alg:improved_fdr}) in a branch of the mCRL2\footnote{www.mcrl2.org} toolset~\cite{BunteGKLNVWWW19} as part of the \emph{ltscompare} tool, which is implemented in \cpp.
As the name of the tool suggests it can be used to check for various preorder and equivalence relations between labelled transition systems.

The data structures used in these implementations compute most concepts, \eg, the antichain membership test and insertion, in the same way. 
However, the implementations of Algorithms~\ref{alg:improved_sfr} and~\ref{alg:improved_fdr} perform the check at line~\ref{alg:improved_sfr:refusals}, or line~\ref{alg:improved_fdr:refusals} respectively, according to the definition of $\refusals$ we presented in Definition~\ref{def:refusals}, whereas the implementations of Algorithms~\ref{alg:original_sfr} and~\ref{alg:original_fdr} compute the refusal check with an additional local search, according to the definition given in~\cite{WangS0LDWL12}, see also Remark~\ref{remark:refusals}.

We first revisit Example~\ref{example:badcase} in Section~\ref{section:badcase}, illustrating that the performance overhead we predict for the original algorithm for checking trace refinement also manifests itself in practice.
In Section~\ref{section:cases}, we then analyse the performance of the algorithms on practical examples consisting of a model of an industrial system and models of concurrent data structures.
Finally, in Section~\ref{section:preprocessing}, we analyse the effect of using a cheap state space minimisation algorithm on the total run time of the algorithms.

All experiments and measurements have been performed on a machine with an Intel Core i7-7700HQ CPU 2.80Ghz and a 16GiB memory limit imposed by \texttt{ulimit -Sv 16777216}.
The source modifications and experiments can be obtained from the downloadable package~\cite{package}.

\subsection{Experiment I: Example~\ref{example:badcase}}
\label{section:badcase}

We have used our implementations of Algorithms~\ref{alg:original_tr} and~\ref{alg:improved_tr} to measure the run time (in seconds) for checking the trace refinement $\lts^k_n \refinedbytrace \lts^k_n$, for all combinations of parameters $n, k \in \{10, 20, \ldots, 500\}$, as described in Example~\ref{example:badcase}.
The results of these measurements are shown as two three-dimensional plots in Figure~\ref{figure:badcase_depth}.

\begin{figure}[H]
\caption{The run time results for Example~\ref{example:badcase} using the depth-first variant of Algorithm~\ref{alg:original_tr} on the left and our Algorithm~\ref{alg:improved_tr} on the right.}\label{figure:badcase_depth}
\vspace{0.5cm}
\includegraphics[width=\textwidth]{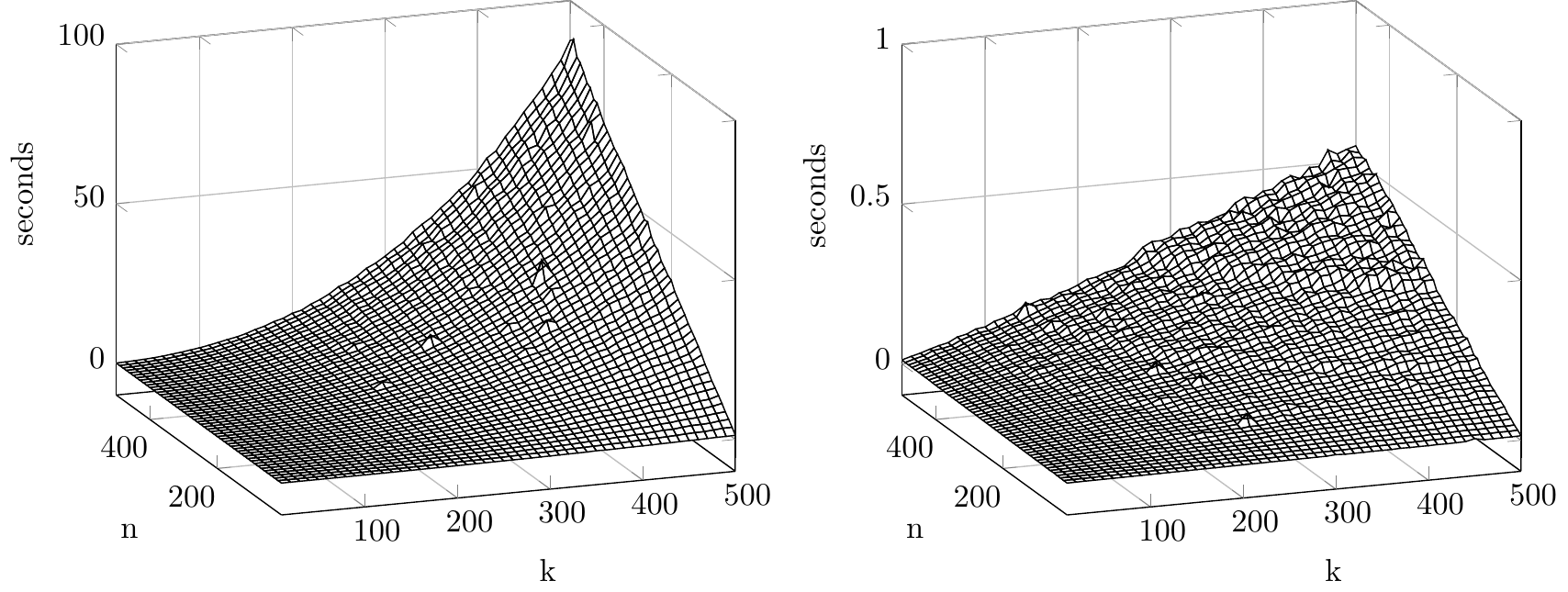}
\end{figure}

These plots show a quadratic growth of Algorithm~\ref{alg:original_tr} in the parameter $k$ and a linear growth in the parameter $n$.
For Algorithm~\ref{alg:improved_tr} the asymptotic growth is linear in both $k$ and $n$.
These observed growths coincide with the analysis that was presented in Example~\ref{example:badcase} for Algorithm~\ref{alg:original_tr} and on page~\pageref{label:analysis} for Algorithm~\ref{alg:improved_tr}.
Note that the scale of the vertical axes of both plots, displaying the run time, differs by two orders of magnitude and the highest runtime (for the $n = 500$ and $k = 500$ case) of Algorithm~\ref{alg:original_tr} is a factor 170 higher than that of Algorithm~\ref{alg:improved_tr}.
As there is no difference in the data structures the difference in run time is entirely due to the different way of inspecting and extending $\working$ and $\antichain$.

The breadth-first variant of Algorithm~\ref{alg:original_tr} was unable to complete the smallest, \ie, $n = k = 10$, case within the given memory limit.
However, as shown in Figure~\ref{figure:badcase_breadth} the run time performance of the breadth-first variant of the improved algorithm is almost equivalent to its depth-first variant.

\begin{figure}[H]
\caption{The run time for Example~\ref{example:badcase} using the breadth-first variant of Algorithm~\ref{alg:improved_tr}.}\label{figure:badcase_breadth}
\vspace{0.5cm}
\begin{center}
\includegraphics[scale=0.9]{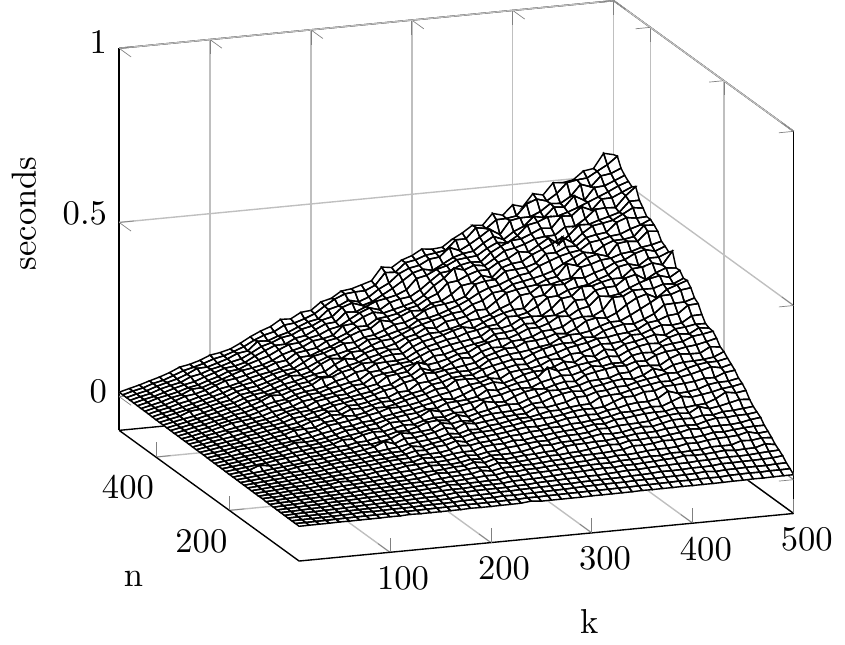}
\end{center}
\end{figure}

\subsection{Experiment II: Practical Examples}
\label{section:cases}

The experiments that we consider are taken from two sources.
First, a model of an industrial system that first exposed the performance issues in
practice of a control system modelled in the Dezyne language~\cite{BeusekomGHHWWW17} .
This example is of a more traditional flavour, in which the
specification is an abstract description of the behaviours at the 
external interface of a control system, and the implementation
is a detailed model that interacts with underlying services to
implement the expected interface.  
For reasons of confidentiality, the industrial model cannot be made available.

Second, we consider several \emph{linearisability tests} of concurrent data structures. 
These models have been taken from~\cite{Paval}, and
consist of six implementations of concurrent data types that, when
\emph{trace refining} their specifications, are guaranteed to be
linearisable. As in~\cite{WangS0LDWL12}, we approximate trace refinement
by the stronger stable failures refinement. 
For these models, the implementation and specification pairs are based on the same descriptions; the difference between the two is that the specification uses a simple construct to guarantee that each method of the concurrent data structure executes atomically.
This significantly reduces the non-determinism and the number of
transitions in the specification models.

In Table~\ref{table:benchmark_datasets} the origin of each model, the number of states and transitions of each implementation and specification LTS, and whether the stable failures refinement relation holds is shown.

\begin{table}[H]
\caption{The number of states and transitions in each benchmark.}\label{table:benchmark_datasets}

\begin{center}
  {
    \renewcommand{\arraystretch}{1.2}
  \small
  \begin{tabularx}{\linewidth}{Xrrrrrr}
  \toprule
  {\bfseries Model} & {\bfseries Ref.} & {\bfseries states spec} & {\bfseries
    trans.\ spec} & {\bfseries $\bm{\refinedbysfr}$} & {\bfseries states impl} &
  {\bfseries trans.\ impl}  \\
  \midrule
  Industrial & - & 24 & 45 & True & 24 551 & 45 447  \\ 
  Coarse set & \cite{Herlihy08} & 50 488 & 64 729 & True & 55 444 & 145 043  \\ 
  Fine-grained set & \cite{Herlihy08} & 3 720 & 3 305 & True & 5 077 & 9 006 \\ 
  Lazy set & \cite{Herlihy08} & 3 565 & 3 980 & True & 24 496 & 41 431  \\ 
  Optimistic set & \cite{Herlihy08} & 25 435 & 28 154 & True & 234 332 & 389 344 \\
  Non-blocking queue & \cite{ShannHC00} & 1 248 & 1 473 & False & 3 030 & 5 799 \\ 
  Treiber stack & \cite{treiber1986systems} & 87 389 & 124 740 & True & 205 634 & 564 862 \\ 
  \bottomrule
\end{tabularx}
  }
\end{center}
\end{table}

\noindent The run time measurements of both Algorithms~\ref{alg:original_sfr} and~\ref{alg:improved_sfr} with both the depth-first and breadth-first variants is shown in Table~\ref{table:runtime_stablefailures}. 
The run times that we report are the averages obtained from five consecutive runs.

\begin{table}[H]
\caption{Run time comparison between Algorithm~\ref{alg:original_sfr} and Algorithm~\ref{alg:improved_sfr} using depth-first (df) and breadth-first (bf) exploration.}\label{table:runtime_stablefailures}

\begin{center}
  {
    \renewcommand{\arraystretch}{1.2}
  \small
  \begin{tabularx}{\linewidth}{Xrrrr}
  \toprule
  {\bfseries Model} & {\bfseries Alg.~\ref{alg:original_sfr} df (s)} &
  {\bfseries Alg.~\ref{alg:original_sfr} bf (s)} & {\bfseries
  Alg.~\ref{alg:improved_sfr} df (s)} & {\bfseries Alg.~\ref{alg:improved_sfr}
  bf (s)} \\ 
    \midrule
  Industrial & 1.36 & 296.29 & 0.15 & 0.17 \\ 
  Coarse set & 9.15 & $\outofmemory$ & 8.61 & 9.06 \\ 
  Fine-grained set & 0.37 & $\outofmemory$ & 0.32 & 0.46 \\ 
  Lazy set & 1.19 & $\outofmemory$ & 1.02 & 1.26 \\ 
  Optimistic set & 16.96 & $\outofmemory$ & 14.13 & 22.67 \\ 
  Non-blocking queue & 0.03 & 0.17 & 0.02 & 0.09 \\ 
  Treiber stack & 148.39 & $\outofmemory$ & 137.52 & 352.59 \\
    \bottomrule
\end{tabularx}
  }
\end{center}
\end{table}

\noindent Here, we observe that the depth-first variant of both algorithms perform similarly with a small run time advantage for Algorithm~\ref{alg:improved_sfr}.
However, for the breadth-first variants our algorithm is able to complete all experiments, whereas Algorithm~\ref{alg:original_sfr} reaches the memory limit, indicated by $\outofmemory$, in five cases and only completes two cases successfully.

To gain more insight into the performance differences between both algorithms we repeat the experiments and report a number of performance metrics.
The reported metrics are the maximum $\working$ size and the number of $\antichain$ membership test that fail (misses), succeed (hits) and the maximum $\antichain$ size during the exploration.
We report the maximum size instead of its size upon termination as these do not necessarily coincide, because inserting an element can evict one or more pairs in $\antichain$.
The following two tables (Tables~\ref{tab:metrics1} and~\ref{tab:metrics2}) show the discussed metrics for the depth-first variant of both algorithms.

\begin{table}[H]
  \caption{Performance metrics for the depth-first variant of Algorithm~\ref{alg:original_sfr}.}
  \label{tab:metrics1}
  \begin{center}
  {
    \renewcommand{\arraystretch}{1.2}
  \small
    \begin{tabularx}{\linewidth}{Xrrrr}
  \toprule
      {\bfseries Model} & {\bfseries $\bm{\working}$ max} & {\bfseries
      $\bm{\antichain}$
      hits} & {\bfseries $\bm{\antichain}$ misses} & {\bfseries $\bm{\antichain}$ max} \\ \midrule
  Industrial & 74 & 36 544 & 43 419 & 43 091 \\ 
  Coarse set & 96 & 93 330 & 58 438 & 55 444 \\ 
  Fine-grained set & 60 & 5 786 & 7 575 & 5 077 \\ 
  Lazy set & 61 & 21 184 & 30 771 & 24 496 \\ 
  Optimistic set & 96 & 234 692 & 354 068 & 238 726 \\ 
  Non-blocking queue & 52 & 548 & 672 & 591 \\ 
  Treiber stack & 101 & 1 238 727 & 756 692 & 234 118 \\ \bottomrule
  \end{tabularx}
}
  \end{center}
\end{table}

\begin{table}[H]
  \caption{Performance metrics for the depth-first variant of Algorithm~\ref{alg:improved_sfr}.}
  \label{tab:metrics2}
  \begin{center}
  {
    \renewcommand{\arraystretch}{1.2}
  \small
    \begin{tabularx}{\linewidth}{Xrrrr}
  \toprule
      {\bfseries Model} & {\bfseries $\bm{\working}$ max} & {\bfseries $\bm{\antichain}$ hits} & {\bfseries $\bm{\antichain}$ misses} & {\bfseries $\bm{\antichain}$ max} \\ \midrule
  Industrial & 69 & 36 369 & 43 090 & 43 091 \\ 
  Coarse set & 96 & 93 330 & 58 438 & 55 444 \\ 
  Fine-grained set & 60 & 5 786 & 7 575 & 5 077 \\ 
  Lazy set & 61 & 21 184 & 30 771 & 24 496 \\ 
  Optimistic set & 96 & 234 692 & 354 068 & 238 728 \\ 
  Non-blocking queue & 43 & 520 & 641 & 634 \\ 
  Treiber stack & 101 & 1 238 727 & 756 692 & 234 119 \\ \bottomrule
  \end{tabularx}
}
  \end{center}
\end{table}

\noindent We observe in Tables~\ref{tab:metrics1} and~\ref{tab:metrics2} that only for the industrial and non-blocking queue models the performance metrics are different.
An explanation for this is that because the antichain membership test is delayed (in Algorithm~\ref{alg:original_sfr}), more pairs are added to $\working$ and these additional pairs increase the number of $\antichain$ checks.
In all other cases, the difference in run time can only be the result of the different refusal computation implementation, as the number of $\antichain$ operations is the same.

The following two tables (Tables~\ref{tab:metrics3} and~\ref{tab:metrics4}) show the obtained performance metrics for the breadth-first variants of both algorithms.
In the experiments where the refinement checking terminates early, due to reaching the memory limit, we report the last observed measurements.

\begin{table}[H]
  \caption{Performance indicators for the breadth-first variant of Algorithm~\ref{alg:original_sfr}.}
  \label{tab:metrics3}
  \begin{center}
  {
    \renewcommand{\arraystretch}{1.2}
  \small
    \begin{tabularx}{\linewidth}{Xrrrr}
  \toprule
      {\bfseries Model} & {\bfseries $\bm{\working}$ max} & {\bfseries $\bm{\antichain}$ hits} & {\bfseries $\bm{\antichain}$ misses} & {\bfseries $\bm{\antichain}$ max} \\ \midrule
  Industrial & 549 263 & 5 459 028 & 12 888 388 & 43 091 \\ 
  Coarse set & 4 710 289 & 13 870 & 7 807 403 & 3 629 \\ 
  Fine-grained set & 6 604 516 & 180 669 & 15 547 890 & 1 900 \\ 
  Lazy set & 6 726 497 & 130 523 & 14 852 835 & 4 306 \\ 
  Optimistic set & 6 366 524 & 38 649 & 14 238 042 & 4 439 \\ 
  Non-blocking queue & 6 262 & 3 078 & 14 560 & 274 \\ 
  Treiber stack & 5 829 902 & 76 114 & 8 340 606 & 4 811 \\ \bottomrule
  \end{tabularx}
}
  \end{center}
\end{table}

\begin{table}[H]
  \caption{Performance indicators for the breadth-first variant of Algorithm~\ref{alg:improved_sfr}.}
  \label{tab:metrics4}
  \begin{center}
  {
    \renewcommand{\arraystretch}{1.2}
  \small
    \begin{tabularx}{\linewidth}{Xrrrr}
  \toprule
      {\bfseries Model} & {\bfseries $\bm{\working}$ max} & {\bfseries $\bm{\antichain}$ hits} & {\bfseries $\bm{\antichain}$ misses} & {\bfseries $\bm{\antichain}$ max} \\ \midrule
  Industrial & 2 243 & 36 369 & 43 090 & 43 091 \\ 
  Coarse set & 3 411 & 96 167 & 60 332 & 55 444 \\ 
  Fine-grained set & 434 & 7 192 & 9 657 & 5 077 \\ 
  Lazy set & 1 748 & 24 340 & 35 192 & 24 496 \\ 
  Optimistic set & 15 209 & 292 525 & 434 218 & 234 352 \\ 
  Non-blocking queue & 338 & 3 426 & 4 032 & 2 675 \\ 
  Treiber stack & 139 218 & 2 411 614 & 1 523 830 & 214 795 \\ \bottomrule
  \end{tabularx}
}
  \end{center}
\end{table}

\noindent From these results it is clear to see that for the breadth-first variant of Algorithm~\ref{alg:original_sfr}, delaying the $\antichain$ insertion of discovered state pairs results in an enormous overhead.
The size of the $\antichain$ remains quite small, which causes many discovered pairs to fail the antichain membership test.
As each pair that fails the membership test is added to $\working$, it causes the $\working$ queue to grow rapidly, until it reaches the memory limit.
On the other hand, for Algorithm~\ref{alg:improved_sfr} we can observe that the number of successful (and unsuccessful) $\antichain$ membership test is quite similar to its depth-first variant.
There can be some differences between these variants as the pairs are discovered in a different order.
The increase of the $\working$ size has the same reason as for ordinary breadth-first search, which depends on the out degree of the visited pairs.
  
To verify that the difference in performance of the depth-first variants is due to the changes of the refusal computation we have implemented another variant of Algorithm~\ref{alg:original_sfr} with the stability check of $\impl$ added.
The run time impact of this change for both depth-first and breadth-first variants of Algorithm~\ref{alg:original_sfr} is shown in Table~\ref{table:benchmarks}.

\begin{table}[H]
\caption{Run time results for Algorithm~\ref{alg:original_sfr} with the stability check of $\impl$.}\label{table:benchmarks}
\begin{center}
  {
    \renewcommand{\arraystretch}{1.2}
\small
  \begin{tabular}{lrr}
  \toprule
    {\bfseries Model} & {\bfseries Alg.~\ref{alg:original_sfr} df (s)} &
    {\bfseries Alg.~\ref{alg:original_sfr} bf (s)} \\ \midrule
  Industrial & 0.17 & 26.32 \\ 
  Coarse set & 9.59 & $\outofmemory$ \\ 
  Fine-grained set & 0.36 & $\outofmemory$ \\ 
  Lazy set & 1.12 & $\outofmemory$ \\ 
  Optimistic set & 15.85 & $\outofmemory$ \\ 
  Non-blocking queue & 0.03 & $\outofmemory$ \\ 
  Treiber stack & 156.67 & $\outofmemory$ \\ \bottomrule
\end{tabular}
  }
\end{center}
\end{table}

\noindent As expected, the run time for this alternative depth-first variant closely matches the run time of the depth-first variant of the improved algorithm.
The alternative breadth-first variant of Algorithm~\ref{alg:original_sfr} is still not able to complete most experiments, but the industrial case has improved quite significantly.
However, the non-blocking queue experiment now reaches the set memory limit.
For this we provide the following explanation.
Note that in case of a failing refinement the exploration stops when a suitable (SF-)witness has been found, which must exist as stable failures refinement does not hold.
Recall that the computation of refusals for (possibly) unstable states as defined in~\cite{WangS0LDWL12}, see Remark~\ref{remark:refusals}, has been implemented using a separate local search for stable states.
We think that in the previous case such an SF-witness was found for an unstable state using this local search.
However, in the alternative version the algorithm continues the exploration of the LTSs when encountering an unstable state (in $\lts_2$), which
causes the $\working$ queue to reach the memory limit.

Finally, we repeat the same experiments while checking for failures-divergences refinement.
This has only been done for Algorithm~\ref{alg:improved_fdr} as the original algorithm for failures-divergences refinement is incorrect.
The run time measurements and the expected result of the failures-divergences refinement check are presented in Table~\ref{table:benchmarks_fdr}.

\begin{table}[H]
\caption{The run time results for checking failures-divergences refinement using Algorithm~\ref{alg:improved_fdr}.}\label{table:benchmarks_fdr}
\begin{center}
\small
  {
    \renewcommand{\arraystretch}{1.2}
  \begin{tabular}{lrrr}
  \toprule
    {\bfseries Model} & {\bfseries Alg.~\ref{alg:improved_fdr} df (s)} &
    {\bfseries Alg.~\ref{alg:improved_fdr} bf (s)} & $\bm{\refinedbyfdr}$ \\ \midrule
  Industrial & 0.05 & 0.05 & False\\ 
  Coarse set & 8.68 & 9.29 & True\\ 
  Fine-grained set & 0.33 & 0.48 & True\\ 
  Lazy set & 1.04 & 1.33  & True\\ 
  Optimistic set & 14.55 & 23.81 & True \\ 
  Non-blocking queue & 0.08 & 0.1 & True \\ 
  Treiber stack & 140.7 & 363.34 & True \\ \bottomrule
\end{tabular}
  }
\end{center}
\end{table}

\noindent The run time results of Table~\ref{table:benchmarks_fdr} show that deciding failures-divergences refinement has a similar performance to deciding stable failures.

\subsection{State Space Minimisation as Preprocessing}\label{section:preprocessing}

The size of the transition systems has a major impact on the practical run time of the refinement checking algorithms we studied, as can also be seen from, \eg,~Tables~\ref{table:benchmark_datasets} and~\ref{table:runtime_stablefailures}.
Note that this is particularly true of the size of the specification LTS, whose normal form can be exponentially larger than the specification itself.
As an alternative to the pruning achieved using antichains, reducing the size of the specification as a preprocessing step to checking for refinement may therefore be an effective tool in improving on the practical run time of these algorithms.
Of course, it is desirable that the computational overhead of the reduction remains minimal.
One possibility is to minimise transition systems using one of the many equivalence relations available for labelled transition systems, see, \eg~\cite{Glabbeek93,BoulgakovGR16}.
When choosing such an equivalence it is important that it has the property that, apart from an appealing run time complexity, the observations, \ie, $\weaktraces$, $\failures$ and $\divergences$, that are extracted from equivalent states are the same.

\emph{Strong} bisimilarity is known to preserve the $\weaktraces$, $\failures$ and $\divergences$ observations of equivalent states.
However, a more substantial state space reduction can often be achieved by considering equivalences that treat the special action $\tau$ as invisible, as the given LTSs often contain $\tau$-transitions.
In~\cite{Roscoe10}, Roscoe suggests to use a variant of weak bisimulation, \viz, \emph{divergence respecting weak bisimulation}, to minimise a transition system.
For \emph{divergence respecting weak} bisimulation it is known that it is a suitable abstraction; see the following theorem~\cite[Theorem 9.2]{Roscoe10}.

\begin{thm}\label{thm:bisimilar}
  Let $\lts = (\states, \init, \transitions)$ be an LTS.
  For two states $s, t \in \states$ that are divergence respecting weak bisimilar it holds that $\weaktraces(s) = \weaktraces(t)$, $\divergences(s) = \divergences(t)$, $\failures(s) = \failures(t)$.
  Hence, also $\failuresbot(s) = \failuresbot(t)$.
\end{thm}

\noindent From a computational point of view, however, (divergence respecting) weak bisimulation is not particularly promising.
For instance, the best known algorithm~\cite{DBLP:journals/iandc/RanzatoT08} for computing weak bisimulation has a worst-case time complexity of $\functionorder(m\cdot n)$, where $n$ is the number of states and $m$ the number of transitions.
Such run time complexities are non-neglible and may result in an undesirable overhead.

In practice, divergence-respecting weak bisimulation often coincides with \emph{divergence preserving branching bisimulation}\footnote{Note that Theorem~6.1 of~\cite{GlabbeekW96}
states that when comparing two systems of which one contains no
$\tau$-transitions, weak and branching bisimilarity coincides.
It is argued that this often applies when comparing specifications with implementations.
This trivially generalises to divergence-respecting weak and divergence preserving branching bisimilarity. 
Also in case of our benchmarks, both relations yield only a minimal difference in size for all models, with the exception of the industrial implementation model.
For that model, however, the costs of computing the weak-bisimulation reduction far exceeds the costs of performing the refinement check.
Perhaps this could be partially mitigated by more efficient (divergence respecting) weak bisimulation algorithms, but that has not been further investigated.}
and it has the far more appealing worst-case run time complexity of $\functionorder(m \cdot \log n)$~\cite{JansenGKW20}, which is equivalent to the run time complexity of computing strong bisimulation~\cite{PaigeT87}.
Divergence-preserving branching bisimulation is stronger than divergence respecting weak bisimulation, \ie, two states that are divergence-preserving branching bisimilar are  also divergence respecting weak bisimilar.
Divergence-preserving branching bisimilarity is, by Theorem~\ref{thm:bisimilar}, therefore a suitable abstraction.

The algorithm that decides divergence-preserving branching bisimulation equivalence between two states can also be adapted to a minimisation procedure with the same $\functionorder(m \cdot \log n)$ time complexity.
We have made the preprocessing step of minimisation modulo divergence-preserving branching bisimulation available as an option in our tool.
In Table~\ref{table:benchmark_dataset_preprocessed} the number of states and transitions of each model of Section~\ref{section:cases}, \emph{after} minimisation modulo divergence-preserving branching bisimulation, and whether the specification and implementation LTSs are divergence-preserving branching bisimilar is shown.

\begin{table}[H]
\caption{The number of states and transitions after diverging preserving branching bisimulation minimisation and whether the LTSs are equivalent in divergence-preserving branching bisimulation semantics denoted by $\branchingbisim$.}\label{table:benchmark_dataset_preprocessed}
\begin{center}
\small
  {
    \renewcommand{\arraystretch}{1.2}
  \begin{tabularx}{\linewidth}{Xrrcrr}
  \toprule
    {\bfseries Model}              & {\bfseries states spec} & {\bfseries
    trans.\ spec} & $\bm{\branchingbisim}$ & {\bfseries states impl} &
    {\bfseries trans.\ impl} \\ \midrule
  Industrial         & 24 & 45 & False & 4 626 & 14 380 \\  
  Coarse set         & 1 089 & 3 618 & True & 1 089 & 3 618 \\ 
  Fine-grained set   & 92 & 210 & True & 92 & 210 \\ 
  Lazy set           & 92 & 210 & True & 92 & 210\\ 
  Optimistic set     & 170 & 410 & True & 170 & 410 \\ 
  Non-blocking queue & 119 & 274 & False & 163 & 378 \\ 
  Treiber stack      & 7 988 & 26 070 & True & 7 988 &  26 070 \\ \bottomrule
\end{tabularx}
  }
\end{center}
\end{table}

\noindent Observe that for most of the models, the minimised implementation and specification LTSs are of equal size; indeed, in those cases the implementation and specification are divergence-preserving branching bisimulation equivalent, so no further stable failures refinement check would be needed.

One option would therefore be to apply the minimisation to both implementation and specification LTSs.
This approach turns out to be beneficial for the Treiber stack example, obtaining a run time of 3 seconds to determine stable failures refinement.
The approach is not beneficial for the other examples.
Moreover, minimising the implementation might even be less effective in case the refinement relation between specification and implementation does not hold, in which case the refinement check will probably quickly determine this fact.

We therefore measure the effect of using minimised specifications, but unmodified implementations.
The run time measurements of checking stable failures refinement using Algorithms~\ref{alg:original_sfr} and~\ref{alg:improved_sfr} using the minimised specification LTS is shown in Table~\ref{table:benchmarks_reduced}.
The time that it takes to compute the divergence-preserving branching bisimulation minimisation is presented in the last column and the other measurements are the run time of the algorithm including preprocessing.

\begin{table}[H]
\caption{Run time comparison between the original algorithm (Algorithm~\ref{alg:original_sfr}) and the improved algorithm (Algorithm~\ref{alg:improved_sfr}) using depth-first (df) and breadth-first (bf) exploration where the specification is reduced modulo divergence-preserving branching bisimulation.}\label{table:benchmarks_reduced}
\begin{center}
\small
  {
    \renewcommand{\arraystretch}{1.2}
    \renewcommand{\tabcolsep}{3pt}
  \begin{tabularx}{\linewidth}{X@{}rrrrr}
  \toprule
    {\bfseries Model} & {\bfseries Alg.~\ref{alg:original_sfr} df (s)} &
    {\bfseries Alg.~\ref{alg:original_sfr} bf (s)} & {\bfseries
    Alg.~\ref{alg:improved_sfr} df (s)} & {\bfseries Alg.~\ref{alg:improved_sfr}
    bf (s)} & {\bfseries Reduction (s)} \\ \midrule
  Industrial & 1.38 & 293.1 & 0.16 & 0.17 & 0.01 \\ 
  Coarse set & 0.74 & $\outofmemory$ & 0.69 & 0.69 & 0.10 \\ 
  Fine-grained set & 0.04 & $\outofmemory$ & 0.04 & 0.04 & 0.01 \\ 
  Lazy set & 0.21 & $\outofmemory$ & 0.15 & 0.15 & 0.01 \\ 
  Optimistic set & 2.52 & $\outofmemory$ & 1.59 & 1.57 & 0.04 \\ 
  Non-blocking queue & 0.02 & 0.04 & 0.02 & 0.02 & 0.01 \\ 
  Treiber stack & 8.19 & $\outofmemory$ & 6.61 & 11.71 & 0.24 \\ \bottomrule
\end{tabularx}
  }
\end{center}
\end{table}

\noindent Comparing these results with Table~\ref{table:runtime_stablefailures} shows that reducing the specification modulo divergence-preserving branching bisimulation can indeed substantially improve the performance of the antichain-based algorithms.
In particular, it never degrades the performance of our algorithms as the preprocessing time is negligible.
For failures-divergences refinement the results, using  Algorithm~\ref{alg:improved_fdr}, are similar, as is shown in Table~\ref{table:benchmarks_fdr_preprocessing}.

\begin{table}[H]
\caption{The run time results for checking failures-divergences refinement using Algorithm~\ref{alg:improved_fdr} where the specification is reduced module divergence-preserving branching bisimulation.}\label{table:benchmarks_fdr_preprocessing}
\begin{center}
\small
  {
    \renewcommand{\arraystretch}{1.2}
  \begin{tabular}{lrr}
  \toprule
    {\bfseries Model} & {\bfseries Alg.~\ref{alg:improved_fdr} df (s)} &
    {\bfseries Alg.~\ref{alg:improved_fdr} bf (s)} \\ \midrule
  Industrial & 0.05 & 0.06 \\ 
  Coarse set & 0.75 & 0.7 \\ 
  Fine-grained set & 0.04 & 0.04 \\ 
  Lazy set & 0.15 & 0.15 \\ 
  Optimistic set & 1.61 & 1.7 \\ 
  Non-blocking queue & 0.02 & 0.02 \\ 
  Treiber stack & 6.76 & 12.13 \\ \bottomrule
\end{tabular}
  }
\end{center}
\end{table}

\section{Conclusion}

Our study of the antichain-based algorithms for deciding trace refinement, stable failures refinement and failures-divergences refinement presented in~\cite{WangS0LDWL12} revealed that the failures-divergences refinement algorithm is incorrect. 
All three algorithms perform suboptimally when implemented using a depth-first search strategy and poorly when implemented using a breadth-first search strategy.
Furthermore, all three algorithms violate the claimed antichain property. 
We propose alternative algorithms for which we have shown correctness and which utilise proper antichains.  
Our experiments indicate significant performance improvements for deciding trace refinement, stable failures refinement and a performance of deciding failures-divergences refinement that is comparable to deciding stable failures refinement. 
We also show that preprocessing using divergence-preserving branching bisimulation offers substantial performance benefits.
The implementation of our algorithms is available in the open source toolset mCRL2~\cite{BunteGKLNVWWW19} and is currently used as the backbone in the commercial F-MDE toolset Dezyne; see also~\cite{BeusekomGHHWWW17}.

\section*{Acknowledgement}
This work is part of the TOP Grants research programme with project number 612.001.751 (AVVA), which is (partly) financed by the Dutch Research Council (NWO).
We also would like to thank the anonymous reviewer for their effort and constructive feedback.

\bibliographystyle{alpha}
\bibliography{bibliography.bib}

\newcommand{\etalchar}[1]{$^{#1}$}
\begin{thebibliography}{vBGH{\etalchar{+}}17}

\bibitem[ACH{\etalchar{+}}10]{AbdullaCHMV10}
P.~A. Abdulla, Y.{-}F. Chen, L.~Hol{\'{\i}}k, R.~Mayr, and T.~Vojnar.
\newblock When simulation meets antichains.
\newblock In J.~Esparza and R.~Majumdar, editors, {\em {TACAS} 2010}, volume
  6015 of {\em {LNCS}}, pages 158--174. Springer, 2010.

\bibitem[BGK{\etalchar{+}}19]{BunteGKLNVWWW19}
O.~Bunte, J.~F. Groote, J.~J.~A. Keiren, M.~Laveaux, T.~Neele, E.~P. de~Vink,
  J.~W. Wesselink, A.~Wijs, and T.~A.~C. Willemse.
\newblock The {mCRL2} toolset for analysing concurrent systems - improvements
  in expressivity and usability.
\newblock In T.~Vojnar and L.~Zhang, editors, {\em {TACAS} 2019}, volume 11428
  of {\em LNCS}, pages 21--39. Springer, 2019.

\bibitem[BGR16]{BoulgakovGR16}
A.~Boulgakov, T.~Gibson{-}Robinson, and A.~W. Roscoe.
\newblock Computing maximal weak and other bisimulations.
\newblock {\em Formal Asp. Comput.}, 28(3):381--407, 2016.

\bibitem[BKO87]{BergstraKO87}
J.~A. Bergstra, J.~W. Klop, and E.{-}R. Olderog.
\newblock Failures without chaos: a new process semantics for fair abstraction.
\newblock In M.~Wirsing, editor, {\em {IFIP} {TC} 2/WG 2.2 1986}, pages
  77--104. North-Holland, 1987.

\bibitem[BR84]{Brookes85}
S.~D. Brookes and A.~W. Roscoe.
\newblock An improved failures model for communicating processes.
\newblock In S.~D. Brookes, A.~W. Roscoe, and Glynn Winskel, editors, {\em
  Seminar on Concurrency 1984}, volume 197 of {\em {LNCS}}, pages 281--305.
  Springer, 1984.

\bibitem[DR10]{DoyenR10}
L.~Doyen and J.{-}F. Raskin.
\newblock Antichain algorithms for finite automata.
\newblock In {\em {TACAS}}, volume 6015 of {\em Lecture Notes in Computer
  Science}, pages 2--22. Springer, 2010.

\bibitem[GABR14]{Gibson-RobinsonABR14}
T.~Gibson{-}Robinson, P.~J. Armstrong, A.~Boulgakov, and A.~W. Roscoe.
\newblock {FDR3} - {A} modern refinement checker for {CSP}.
\newblock In E.~{\'{A}}brah{\'{a}}m and K.~Havelund, editors, {\em {TACAS}
  2014}, volume 8413 of {\em {LNCS}}, pages 187--201. Springer, 2014.

\bibitem[GABR16]{Gibson-Robinson16}
T.~Gibson{-}Robinson, P.~J. Armstrong, A.~Boulgakov, and A.~W. Roscoe.
\newblock {FDR3:} a parallel refinement checker for {CSP}.
\newblock {\em {STTT}}, 18(2):149--167, 2016.

\bibitem[GB16]{GomesB16}
A.~O. Gomes and A.~Butterfield.
\newblock Modelling the haemodialysis machine with circus.
\newblock In M.~J. Butler, K.{-}D. Schewe, A.~Mashkoor, and M.~Bir{\'{o}},
  editors, {\em {ABZ} 2016}, volume 9675 of {\em LNCS}, pages 409--424.
  Springer, 2016.

\bibitem[GBC{\etalchar{+}}17]{Gibson-Robinson17}
T.~Gibson{-}Robinson, G.~H. Broadfoot, G.~Carvalho, P.~J. Hopcroft, G.~Lowe,
  S.~Nogueira, C.~O'Halloran, and A.~Sampaio.
\newblock {FDR:} from theory to industrial application.
\newblock In {\em Concurrency, Security, and Puzzles}, volume 10160 of {\em
  Lecture Notes in Computer Science}, pages 65--87. Springer, 2017.

\bibitem[Hoa85]{Hoare85}
C.~A.~R. Hoare.
\newblock {\em Communicating Sequential Processes}.
\newblock Prentice-Hall, 1985.

\bibitem[HS08]{Herlihy08}
M.~Herlihy and N.~Shavit.
\newblock {\em The art of multiprocessor programming}.
\newblock Morgan Kaufmann, 2008.

\bibitem[JGKW20]{JansenGKW20}
D.~N. Jansen, J.~Friso Groote, J.~J.~A. Keiren, and A.~Wijs.
\newblock An {O(m log n)} algorithm for branching bisimilarity on labelled
  transition systems.
\newblock In {\em {TACAS} {(2)}}, volume 12079 of {\em Lecture Notes in
  Computer Science}, pages 3--20. Springer, 2020.

\bibitem[KS90]{KanellakisS90}
P.~C. Kanellakis and S.~A. Smolka.
\newblock {CCS} expressions, finite state processes, and three problems of
  equivalence.
\newblock {\em Inf. Comput.}, 86(1):43--68, 1990.

\bibitem[Lav19]{package}
M.~Laveaux.
\newblock Downloadable sources and benchmarks for the experimental validation.
\newblock 2019.
\newblock \url{https://doi.org/10.5281/zenodo.3449420}.

\bibitem[LGW19]{LaveauxGW19}
M.~Laveaux, J.~F. Groote, and T.~A.~C. Willemse.
\newblock Correct and efficient antichain algorithms for refinement checking.
\newblock In {\em {FORTE}}, volume 11535 of {\em Lecture Notes in Computer
  Science}, pages 185--203. Springer, 2019.

\bibitem[Pav18]{Paval}
R.~Paval.
\newblock Modeling and verifying concurrent data structures.
\newblock Master's thesis, Eindhoven University of Technology, 2018.
\newblock https://research.tue.nl/files/93882157/Thesis\_Roxana\_Paval.pdf.

\bibitem[PT87]{PaigeT87}
Robert Paige and Robert~Endre Tarjan.
\newblock Three partition refinement algorithms.
\newblock {\em {SIAM} J. Comput.}, 16(6):973--989, 1987.

\bibitem[Ros94]{Roscoe94}
A.~W. Roscoe.
\newblock Model-checking {CSP}.
\newblock In A.~W. Roscoe, editor, {\em A Classical Mind: essays in Honour of
  C.\,A.\,R. Hoare}, chapter~21, pages 353--378. Prentice Hall International
  (UK) Ltd., 1994.

\bibitem[Ros10]{Roscoe10}
A.~W. Roscoe.
\newblock {\em Understanding Concurrent Systems}.
\newblock Texts in Computer Science. Springer, 2010.

\bibitem[RT08]{DBLP:journals/iandc/RanzatoT08}
F.~Ranzato and F.~Tapparo.
\newblock Generalizing the paige-tarjan algorithm by abstract interpretation.
\newblock {\em Inf. Comput.}, 206(5):620--651, 2008.

\bibitem[SHC00]{ShannHC00}
C.{-}H. Shann, T.{-}L. Huang, and C.~Chen.
\newblock A practical nonblocking queue algorithm using compare-and-swap.
\newblock In {\em {ICPADS} 2000}, pages 470--475. {IEEE} Computer Society,
  2000.

\bibitem[Tre86]{treiber1986systems}
R.~K. Treiber.
\newblock {\em Systems programming: Coping with parallelism}.
\newblock International Business Machines Incorporated, Thomas J. Watson
  Research, 1986.

\bibitem[vBGH{\etalchar{+}}17]{BeusekomGHHWWW17}
R.~van Beusekom, J.~F. Groote, P.~F. Hoogendijk, R.~Howe, J.~W. Wesselink,
  R.~Wieringa, and T.~A.~C. Willemse.
\newblock Formalising the {Dezyne} modelling language in {mCRL2}.
\newblock In L.~Petrucci, C.~Seceleanu, and A.~Cavalcanti, editors, {\em
  {FMICS-AVoCS} 2017}, volume 10471 of {\em {LNCS}}, pages 217--233. Springer,
  2017.

\bibitem[vG93]{Glabbeek93}
R.~J. van Glabbeek.
\newblock The linear time - branching time spectrum {II}.
\newblock In E.~Best, editor, {\em {CONCUR} 1993}, volume 715 of {\em {LNCS}},
  pages 66--81. Springer, 1993.

\bibitem[vG17]{Glabbeek17}
R.~J. van Glabbeek.
\newblock A branching time model of {CSP}.
\newblock In T.~Gibson{-}Robinson, P.~J. Hopcroft, and R.~Lazic, editors, {\em
  Concurrency, Security, and Puzzles - Essays Dedicated to Andrew William
  Roscoe on the Occasion of His 60th Birthday}, volume 10160 of {\em {LNCS}},
  pages 272--293. Springer, 2017.

\bibitem[vG19]{Glabbeek19}
R.~J. van Glabbeek, 2019.
\newblock Personal Communication, 7 January 2019.

\bibitem[vGLT09]{GlabbeekLT09}
R.~J. van Glabbeek, B.~Luttik, and N.~Tr\u{c}ka.
\newblock Branching bisimilarity with explicit divergence.
\newblock {\em Fundam. Inform.}, 93(4):371--392, 2009.

\bibitem[vGW96]{GlabbeekW96}
R.~J. van Glabbeek and W.~P. Weijland.
\newblock Branching time and abstraction in bisimulation semantics.
\newblock {\em J. {ACM}}, 43(3):555--600, 1996.

\bibitem[WDHR06]{WulfDHR06}
M.~De Wulf, L.~Doyen, T.~A. Henzinger, and J.{-}F. Raskin.
\newblock Antichains: {A} new algorithm for checking universality of finite
  automata.
\newblock In Thomas Ball and Robert~B. Jones, editors, {\em Computer Aided
  Verification, 18th International Conference, {CAV} 2006, Seattle, WA, USA,
  August 17-20, 2006, Proceedings}, volume 4144 of {\em Lecture Notes in
  Computer Science}, pages 17--30. Springer, 2006.

\bibitem[WSS{\etalchar{+}}12]{WangS0LDWL12}
T.~Wang, S.~Song, J.~Sun, Y.~Liu, J.~S. Dong, X.~Wang, and S.~Li.
\newblock More anti-chain based refinement checking.
\newblock In T.~Aoki and K.~Taguchi, editors, {\em {ICFEM}}, volume 7635 of
  {\em {LNCS}}, pages 364--380. Springer, 2012.

\end{thebibliography}

\newpage
\begin{appendix}

\section{Proof of Lemma~\ref{lemma:subset_sfr}}

\renewcommand{\thethm}{\ref{lemma:subset_sfr}} 

\lemmasubsetsfr*
\begin{proof}
  We use induction on the length of sequences in $\actions^*$ to prove the statement.

  Base case.
  First, $\init' \tracetransition{\emptytrace}' U$ iff $U = \init'$ by definition.
  The state $\init'$ is equal to $\{s \mid \init \weaktransition{\emptytrace} s\}$ as defined in the normalisation.
  Hence, for every state $s \in \tostates{\init'}_\lts$ we have $\init \weaktransition{\emptytrace} s$.

  Inductive case.
  Pick any sequence $\weaktrace \in \actions^*$ of length $i$ and suppose that the statement holds for all sequences in $\actions^*$ of length $i$.
  Take an arbitrary state $V \in \states'$ and action $a \in \actions$ such that $\init' \tracetransition{\weaktrace\, a}' V$.
  Then there is a state $U \in \states'$ such that $\init' \tracetransition{\weaktrace}' U$ and $U \transition{a}' V$.
  By definition of normalisation there is a transition $U \transition{a}' V$ if and only if $V = \{t \in \states \mid \exists s \in U : s \weaktransition{a} t\}$.
  So for all states $t \in \tostates{V}_\lts$ there is a state $s \in \tostates{U}_\lts$ such that $s \weaktransition{a} t$.
  By the induction hypothesis it holds that for all $s \in \tostates{U}_\lts$ there is a weak transition $\init \weaktransition{\weaktrace} s$.
  But then we may conclude that $\init \weaktransition{\weaktrace\,a} t$ for all $t \in \tostates{V}_\lts$.
\end{proof}

\section{Proof of Lemma~\ref{lemma:subset_existence_sfr}}

\renewcommand{\thethm}{\ref{lemma:subset_existence_sfr}} 

\lemmasubsetexistencesfr*
\begin{proof}
  We proceed using an induction on the length of sequences in $\actions^*$.

  Base case. Let $s \in \states$ and suppose $\init \weaktransition{\emptytrace} s$.
  Then $s \in \tostates{\init'}_\lts$, since $\init'$ is defined as $\{s \in \states \mid \init \weaktransition{\emptytrace} s\}$.
  Moreover, we trivially have $\init' \tracetransition{\emptytrace}' \init'$.

  Inductive step.
  Pick any sequence $\weaktrace \in \actions^*$ of length $i$ and suppose that the statement holds for all sequences in $\actions^*$ of length $i$.
  Take an arbitrary state $t \in \states$ and action $a \in \actions$ such that $\init \weaktransition{\weaktrace\,a} t$.
  Then there is a state $s \in \states$ such that $\init \weaktransition{\weaktrace} s$ and $s \weaktransition{a} t$. Fix such a state $s \in \states$.
  From the induction hypothesis it then follows that there is a state $U \in \states'$ such that $s \in \tostates{U}_\lts$ and $\init' \tracetransition{\weaktrace}' U$. Fix this state $U \in \states'$.
  Let $V$ be equal to $\{t \in \states \mid \exists u \in U : u \weaktransition{a} t\}$; then by definition, $U \transition{a}' V$.
  It follows that $\init' \tracetransition{\weaktrace\,a}' V$.
  Finally, $t \in \tostates{V}_\lts$ follows from $s \weaktransition{a} t$ and $s \in \tostates{U}_\lts$.
\end{proof}

\section{Proof of Lemma~\ref{lemma:not_weaktrace_sfr}}

\renewcommand{\thethm}{\ref{lemma:not_weaktrace_sfr}} 

\lemmanotweaktracesfr*
\begin{proof}
  \leavevmode
  \begin{itemize}[align=left,beginpenalty=999]

  \item[$(\Longrightarrow$)]
  We use induction on the length of the sequences in $\actions^*$.

  \noindent Base case. The implication holds vacuously since $\emptytrace \in \weaktraces(\lts)$.

  \noindent Inductive step.
  Pick a sequence $\weaktrace \in \actions^*$ of length $i$ and suppose that the implication holds for all sequences of length $i$.
  Assume an arbitrary action $a \in \actions$ such that $\weaktrace\,a \notin \weaktraces(\lts)$.
  From $\weaktrace\,a \notin \weaktraces(\lts)$ it follows that there is no state $t \in \states$ such that $\init \weaktransition{\weaktrace\,a} t$.
  We distinguish two cases:
  \begin{itemize}
  \item Case $\weaktrace \notin \weaktraces(\lts)$.
  From the induction hypothesis we obtain $\init' \tracetransition{\weaktrace}' \emptyset$ and $\emptyset \transition{a}' \emptyset$ by definition.
  We may therefore also conclude $\init' \tracetransition{\weaktrace\,a}' \emptyset$.

  \item Case $\weaktrace \in \weaktraces(\lts)$.
  Then there is a state $s \in \states$ such that $\init \weaktransition{\weaktrace} s$.
  Since $\lts'$ is deterministic there is a unique state $U \in \states'$ such that both $s \in \tostates{U}_\lts$ and $\init' \tracetransition{\weaktrace}' U$ by Lemmas~\ref{lemma:subset_existence_sfr} and~\ref{lemma:deterministic}.
  For all states $u \in \states$ satisfying $\init \weaktransition{\weaktrace} u$ (which exist as $\weaktrace \in \weaktraces(\lts)$) there cannot be a state $t \in \states$ such that $u \weaktransition{a} t$ by the observation that $\weaktrace\,a \notin \weaktraces(\lts)$.
  Therefore, $U \transition{a}' \emptyset$ and thus also $\init' \tracetransition{\weaktrace\,a}' \emptyset$.
  \end{itemize}

  \item[($\Longleftarrow$)]
  Suppose $\init' \tracetransition{\weaktrace}' \emptyset$.
  Towards a contradiction, assume that $\weaktrace \in \weaktraces(\lts)$. Then there is a state $s \in \states$ such that $\init \weaktransition{\weaktrace} s$.
  By Lemma~\ref{lemma:subset_existence_sfr}, there must be some $U \in \states'$ such that $s \in \tostates{U}_\lts$ and $\init' \tracetransition{\weaktrace}' U$.
  Since $\lts'$ is deterministic, by Lemma~\ref{lemma:deterministic}, we obtain that this state $U$ must be such that $U = \emptyset$.\qedhere
  \end{itemize}
\end{proof}

\section{Proof of Lemma~\ref{lemma:subset_existence_fdr}}

\renewcommand{\thethm}{\ref{lemma:subset_existence_fdr}} 

\lemmasubsetexistencefdr*
\begin{proof}
  Proof by induction on the length of sequences that are not divergences, or minimal divergences.
    
  Base case. 
  The empty trace $\emptytrace$ satisfies $\emptytrace \notin \divergences(\lts)$ or $\emptytrace \in \divergencesbot(\lts)$ by definition.
  Hence, we must show that for all states $s \in \states$ satisfying $\init \weaktransition{\emptytrace} s$ there is a state $U \in \states'$ such that $s \in \tostates{U}_\lts$ and $\init' \tracetransition{\emptytrace}' U$.
  We know that if $\init \weaktransition{\emptytrace} s$ then $s \in \tostates{\init'}_\lts$, because $\init'$ is defined as $\{s \in \states \mid \init \weaktransition{\emptytrace} s\}$ in the normalisation.
  Finally, we also know that $\init' \tracetransition{\emptytrace}' \init'$.
    
  Inductive step.
  Suppose that the statement holds for all sequences of length $i$ that are either not divergences or minimal divergences of $\lts$.
  Pick a sequence $\weaktrace \in \actions^*$ of length $i$, an arbitrary state $t \in \states$ and action $a \in \actions$ such that $\init \weaktransition{\weaktrace\,a} t$ and $\weaktrace\,a \notin \divergences(\lts)$ or $\weaktrace\,a \in \divergencesbot(\lts)$.
  Note that whenever $\weaktrace\,a \notin \divergences(\lts)$ or $\weaktrace\,a \in \divergencesbot(\lts)$ then $\weaktrace \notin \divergences(\lts)$.
  From $\init \weaktransition{\weaktrace\,a} t$ it follows that there is a state $s \in \states$ such that $\init \weaktransition{\weaktrace} s$ and $s \weaktransition{a} t$.
  By our induction hypothesis it then follows that there is a state $U \in \states'$ such that $s \in \tostates{U}_\lts$ and $\init' \tracetransition{\weaktrace}' U$.
  For all states $u \in \tostates{U}_\lts$ it holds that $\init \weaktransition{\weaktrace} u$ by Lemma~\ref{lemma:subset_fdr}, so by definition of divergences it must be that $\diverges{u}$ does not hold and hence $\diverges{\tostates{U}_\lts}$ does not hold.
  Let $V$ be equal to $\{v \in \states \mid \exists u \in U : u \weaktransition{a} v\}$ such that $U \transition{a}' V$ by definition of the normalisation.
  It follows that $\init' \tracetransition{\weaktrace\,a}' V$.
  Finally, $t \in \tostates{V}_\lts$ follows from $s \weaktransition{a} t$ and $s \in \tostates{U}_\lts$.
\end{proof}

\section{Proof of Lemma~\ref{lemma:not_divergence_fdr}}

\renewcommand{\thethm}{\ref{lemma:not_divergence_fdr}} 

\lemmanotdivergencefdr*
\begin{proof}
  Let $\weaktrace \in \actions^*$ and $U \in \states'$ be such that $\init' \tracetransition{\weaktrace}' U$ and not $\diverges{\tostates{U}_\lts}$.
  Towards a contradiction, assume that $\weaktrace \in \divergences(\lts)$.
  We distinguish two cases:
  \begin{itemize}
  \item $\weaktrace \in \divergencesbot(\lts)$. 
        Let $t \in \states$ be such that $\init \weaktransition{\weaktrace} t$ and $\diverges{t}$.
        Note that due to the determinism of $\normalize(\lts)$ and Lemma~\ref{lemma:subset_existence_fdr}, for all $s \in \states$ such that $\init \weaktransition{\weaktrace} s$, we have $s \in U$.
        Hence also $t \in U$. But $\diverges{t}$ then implies $\diverges{\tostates{U}_\lts}$.
        Contradiction.
  \item $\weaktrace \notin \divergencesbot(\lts)$.
        Then $\weaktrace = \weaktrace' \weaktrace''$ for some $\weaktrace' \in \divergencesbot(\lts)$.
        Let $t \in\states$ be such that $\init \weaktransition{\weaktrace'} t$ and $\diverges{t}$.
        Then, by Lemma~\ref{lemma:subset_existence_fdr}, there must be some $V \in \states'$ such that $\init' \tracetransition{\weaktrace'}' V$ and $t \in V$. Let $V$ be such.
        Since $\diverges{t}$ and $t \in V$, $V$ has no outgoing transitions in $\normalize(\lts)$. 
        In particular, we cannot have $V \tracetransition{\weaktrace''}' U$, and because $\normalize(\lts)$ is deterministic, we also cannot have $\init' \tracetransition{\weaktrace}' U$. Contradiction.\qedhere
  \end{itemize}
\end{proof}

\section{Proof of Lemma~\ref{lemma:not_weaktrace_fdr}}

\renewcommand{\thethm}{\ref{lemma:not_weaktrace_fdr}} 

\lemmanotweaktracefdr*
\begin{proof}
  \leavevmode
  \begin{itemize}[align=left,beginpenalty=999]

  \item[$(\Longrightarrow$)]
  Proof by induction on the length of sequences in $\actions^*$.
  
  \noindent Base case.
  The implication holds vacuously since $\emptytrace \in \weaktraces(\lts)$.
  
  \noindent Inductive step.
  Suppose that the statement holds for all sequences $\actions^*$ of length $i$.
  Pick a sequence $\weaktrace \in \actions^*$ of length $i$.
  Take an arbitrary action $a \in \actions$ such that $\weaktrace\,a \notin (\divergences(\lts) \cup \weaktraces(\lts))$.
  From $\weaktrace\,a \notin \weaktraces(\lts)$ it follows that there is no state $t \in \states$ such that $\init \weaktransition{\weaktrace\,a} t$.
  From $\weaktrace\,a \notin \divergences(\lts)$ it follows that $\weaktrace \notin \divergences(\lts)$.
  Now, there are two cases to distinguish:
  
  \begin{itemize}
  \item Case $\weaktrace \notin \weaktraces(\lts)$.
    From the induction hypothesis we obtain $\init' \tracetransition{\weaktrace}' \emptyset$ and $\emptyset \transition{a}' \emptyset$ by definition.
    Thus $\init' \tracetransition{\weaktrace\,a}' \emptyset$.
  
  \item Case $\weaktrace \in \weaktraces(\lts)$.
    There is a state $s \in \states$ such that $\init \weaktransition{\weaktrace} s$.
    Since $\lts'$ is deterministic there is a unique state $U \in \states'$ such that $s \in \tostates{U}_\lts$ and $\init' \tracetransition{\weaktrace}' U$ by Lemma~\ref{lemma:subset_existence_fdr} and \ref{lemma:deterministic}.
    We may furthermore conclude that $\diverges{\tostates{U}_\lts}$ does not hold.
    Because $\weaktrace\,a \notin \weaktraces(\lts)$, no state $u \in \states$ for which $\init \weaktransition{\weaktrace} u$ satisfies $u \weaktransition{a} t$, for any state $t$.
    Therefore, by definition, $U \transition{a}' \emptyset$, and thus $\init' \tracetransition{\weaktrace\,a}' \emptyset$.
  \end{itemize}
  
  \item[$(\Longleftarrow$)]
  Suppose $\init' \tracetransition{\weaktrace}' \emptyset$.
		  From the observation that $\diverges{\tostates{\emptyset}_\lts}$ does not hold and Lemma~\ref{lemma:not_divergence_fdr} it follows that $\weaktrace \notin \divergences(\lts)$.  
  Towards a contradiction, assume that $\weaktrace \in \weaktraces(\lts)$. 
  Then there is a state $s \in \states$ such that $\init \weaktransition{\weaktrace} s$.
  By Lemma~\ref{lemma:not_weaktrace_fdr}, there must be some $U \in \states'$ such that $s \in \tostates{U}_\lts$ and $\init' \tracetransition{\weaktrace}' U$.
  Since $\lts'$ is deterministic, by Lemma~\ref{lemma:deterministic}, we obtain that this state $U$ must be such that $U = \emptyset$.
  Contradiction.\qedhere
  \end{itemize}
\end{proof}

\end{appendix}

\end{document}